\PassOptionsToPackage{table,xcdraw}{xcolor}
\makeatletter
\providecommand*{\input@path}{}
\g@addto@macro\input@path{{.}{include/}{../include/}}
\makeatother

\newif\ifspringer
\springerfalse

\newif\ifec
\ecfalse

\newif\ifanon
\anonfalse

\newif\ifjair
\jairfalse

\ifspringer
    \RequirePackage{amsmath}
    \documentclass[smallextended,envcountsame,nospthms]{svjour3}     %
    \smartqed  %
    \usepackage{fairness}
    \usepackage{tikz}
    \usepackage[dvipsnames]{xcolor}
    \journalname{}
    
    \institute{A. M. Kazachkov
                \email{akazachkov@ufl.edu}%
    }
\else
    \ifec
        \documentclass[format=acmsmall, review=true]{acmart}
        \usepackage{acm-ec-24}
        \usepackage{fairness}
        \usepackage{tikz}

        \setcitestyle{authoryear}
    \else
        \ifjair
            \documentclass[twoside,11pt]{article}
            \usepackage{fairness}
            
            \usepackage{jair, theapa, rawfonts}
            \jairheading{x}{2024}{1-24}{6/24}{xx/xx}
            \usepackage[none]{hyphenat}
        \else
            \documentclass[12pt]{article} %
            \usepackage{fairness}
            \usepackage[none]{hyphenat}
            
            \newenvironment{acks}{\paragraph{\footnotesize Acknowledgments.}\itshape\footnotesize}{\par}
    
            \ifanon
            \usepackage{setspace}
            \spacing{1.5}
        \fi

    \fi
\fi

\ifanon
    \newcommand{\citepar}[1]{\citep{#1}}
    \newcommand{\citeaut}[1]{\cite{#1}}
\else
    \ifjair
        \newcommand{\citeaut}[1]{\citeA{#1}}
        \newcommand{\citepar}[1]{\shortcite{#1}}
    \else
        \newcommand{\citepar}[1]{\cite{#1}}
        \newcommand{\citeaut}[1]{\citet{#1}}
    \fi
\fi

\title{Symmetrically Fair Allocations of Indivisible Goods}
\ifanon
\author{Authors Redacted for Double-Blind Review}
\else
    \ifjair
        \author{\name Connor Johnston \email cjohnston1@ufl.edu \\
               \name Aleksandr M. Kazachkov \email akazachkov@ufl.edu \\
               \addr University of Florida,\\
               Gainesville, FL  32605 USA
               }
    \else
        \author{
            Connor Johnston
            \and
            Aleksandr M. Kazachkov
        }
    \fi
\fi
\date{\today{}}

\newcommand{\genericallocation}[1]{\mathcal{#1}} %
\newcommand{\genericbundle}[2]{#1_{#2}} %

\newcommand{\alloc}[1][A]{\genericallocation{#1}}
\newcommand{\bundle}[1]{\genericbundle{A}{#1}} %

\newcommand{\itemvalue}[2]{v_{#1 #2}} %
\newcommand{\bundlevalue}[2]{v_{#1}(#2)} %
\newcommand{\maxItemValAgentBundle}[2]{\bar{v}_{#1}({#2})} %
\newcommand{\minItemValAgentBundle}[2]{\underline{v}_{#1}({#2})} %

\newcommand{\nTuples}{\mathcal{T}}
\newcommand{\nTuplesForAgent}[1]{\nTuples^{#1}}
\newcommand{\nTupleElementForAgent}[2]{\nTuplesForAgent{#1}_{#2}}

\newcommand{\itemallocated}{\mathrm{item\_allocated}}
\newcommand{\True}{\mathrm{True}}
\newcommand{\False}{\mathrm{False}}

\newcommand{\player}[1]{#1}

\newcommand{\toOne}[1]{{\color{purple} \fbox{#1}}}
\newcommand{\toTwo}[1]{{\color{blue!50!black} \fbox{#1}}}

\begin{document}
\ifspringer
\renewcommand\subclassname{{\bfseries Mathematics Subject Classification}\enspace}
\fi

\ifec
    \begin{titlepage}
\else
\maketitle
\fi

\begin{abstract}
We consider allocating indivisible goods with provable fairness guarantees that are satisfied regardless of which bundle of items each agent receives. Symmetrical allocations of this type are known to exist for divisible resources, such as consensus splitting of a cake into parts, each having equal value for all agents, ensuring that in any allocation of the cake slices, no agent would envy another. For indivisible goods, one analogous concept relaxes envy freeness to guarantee the existence of an allocation in which any bundle is worth as much as any other, up to the value of a bounded number of items from the other bundle. Previous work has studied the number of items that need to be removed. In this paper, we improve upon these bounds for the specific setting in which the number of bundles equals the number of agents.
    
Concretely, we develop the theory of symmetrically envy free up to one good, or symEF1, allocations.
We prove that a symEF1 allocation exists if the vertices of a related graph can be partitioned (colored) into as many independent sets as there are agents.
This sufficient condition always holds for two agents, and for agents that have identical, disjoint, or binary valuations.
We further prove conditions under which exponentially-many distinct symEF1 allocations exist.
Finally, we perform computational experiments to study the incidence of symEF1 allocations as a function of the number of agents and items when valuations are drawn uniformly at random.
\end{abstract}

\ifec
\maketitle
    \end{titlepage}
\else
\fi

\section{Introduction \& Preliminaries}
\label{sec:intro}

Suppose a humanitarian organization distributes donations to community members (referred to as \emph{agents}) with heterogeneous preferences, such as during disaster relief efforts.
To facilitate logistics, the organization prepackages boxes of items that it does \emph{not} assign to specific agents, but rather agents may receive an arbitrary bundle.
For example, the boxes might be set out on a table for decentralized pickups, or they may arrive as part of an aid airdrop.
Afterwards, each agent can observe all other agents' bundles and evaluates the \emph{fairness} of the outcome, such as experiencing \emph{envy} if they value another bundle more than their own.
Our central research question is whether there exists a partition the items into bundles that \emph{all} appear fair to each agent, regardless of which bundle each agent receives.

Specifically, we study the nonwasteful allocation of indivisible goods among agents with known additive nonnegative valuations, under \emph{fairness} and \emph{symmetry} constraints.
We adopt the common fairness standard of \emph{envy freeness up to one good} (EF1)~\citepar{LipMar04_EF1-implicit},
which enforces that the maximum envy (additional value) any agent has for a bundle they do not receive is at most the value of their favorite item in this other bundle.
An EF1 allocation can always be efficiently found, such as if agents pick goods in round-robin fashion.
Imposing symmetry is a much stronger requirement than EF1:
we seek a set of bundles such that any assignment of the bundles to the agents is an EF1 allocation.
We call such a partition \emph{symmetrically envy free up to one good}, or \emph{symEF1}.

A symmetrically fair allocation is intuitively appealing.
One practical motivation comes from the humanitarian logistics setting mentioned above, for which there is currently a lack of provable fairness guarantees.
As another example, consider designing a video game in which each agent chooses a character possessing traits such as health, speed, and strength.
To improve game play, the set of available characters should be roughly evenly matched, regardless of which character is picked.

A symEF1 partition also has theoretical significance,
through the property that any permutation of the bundles among the agents is an EF1 allocation.
This has applications to group fairness~\citepar{GoldHoll22_consensus-k-splitting,ManSukS22_general-symEFk} and fair division over time~\citepar{BenKazProPso18,BenKazProPsoZen23+_fair-efficient-online-allocations}.
In particular, in the latter online context, when one item arrives at each time period, an algorithm that provides tight envy guarantees is simply assigning the item to an agent uniformly at random.
This idea does not directly extend to the setting in which the batch size in each period is more than one.
However, if each period admits a symEF1 partition, then selecting a random permutation of the bundles is effectively randomizing over EF1 allocations,
which leads to the a simpler proof technique for the batch case, compared to that of \citeaut{BenKazProPso18}.

\subsection{Contributions}
\label{sec:contributions}
We formally define symEF1 allocations in \cref{sec:notation}, and we connect this concept to related literature in \cref{sec:related-literature}.
Since symEF1 is a much stronger condition than EF1, existence of such allocations might seem unattainable in general.
To assuage this natural skepticism, we lead with positive, intuitive special cases in \cref{sec:existence-examples}.
However, we confirm in \cref{sec:not-guaranteed} that there exist instances for which the symEF1 condition is unsatisfiable.

Nevertheless, we return to a positive outlook in \cref{sec:symEF1-sufficiency}, containing our main results.
We prove a sufficient condition for the existence of a symEF1 allocation via a vertex coloring problem on an auxiliary graph constructed from agent valuations.
This implies \cref{cor:n2-exists}, that a \emph{balanced} symEF1 allocation \emph{always} exists for two agents.
We then show that symEF1 allocations are not isolated occurrences:
even with two agents and four items, at least two distinct allocations are guaranteed to be symEF1.
Further, we employ our graph construction to identify \emph{exponentially-many} symEF1 allocations.

In \cref{sec:comparision-with-existing-criteria}, we briefly consider how enforcing symmetry relates to other fairness criteria, particularly highlighting differences with solutions that satisfy envy freeness up to any good and maximum Nash welfare.
Finally, in \cref{sec:simulation}, we present computational results with a simulation study on the incidence of symEF1 allocations as a function of the number of agents and items.
Our results inspire several conjectures that we state in \cref{sec:conclusion}.

\subsection{Notation}
\label{sec:notation}

We assume that there are $n$ agents and $m$ items.
For each item $j \in [m] \defeq \{1,\ldots,m\}$, agent~$i \in [n]$ has a fixed, known value $\itemvalue{i}{j} \in [0,1]$.
We assume that valuations are \emph{additive}: for any bundle of items $A \subseteq [m]$, the value agent~$i$ has for $A$ is 
    $\bundlevalue{i}{A} \defeq \sum_{j \in A} \itemvalue{i}{j}$.
An \emph{instance} of our problem refers to a fixed $n$, $m$, and set of valuations $\{v_{ij}\}_{i \in [n], j\in [m]}$.

A \emph{partition} is a set of $n$ \emph{bundles}, $\alloc = (\bundle{k})_{k \in [n]}$, in which each item $j \in [m]$ is contained in precisely one bundle.
An \emph{assignment} of the bundles to agents matches each bundle to an agent such that each agent receives exactly one bundle.
An \emph{allocation} typically refers to a partition and particular assignment.

Informally, a nonsymmetric fairness axiom is achievable for an instance when there exists a partition and specific assignment in which all agents agree that their assigned items are fair with respect to this axiom.
The symmetric version of the same fairness condition requires that the property holds for all possible assignments of bundles to agents.
For existence results under this kind of symmetry,
the assignment component of an allocation is immaterial,
and so we may say ``allocation'' without specifying who receives the bundles.

One common fairness criterion is \emph{envy freeness}~\citepar{Fol67_resource-allocation-EF},
in which each agent weakly prefers the bundle they are assigned over any other agent's bundle; i.e.,
    $
        \bundlevalue{i}{\bundle{i}} \ge \bundlevalue{i}{\bundle{k}}
    $
for all $i,k \in [n]$.
The symmetric version requires that
    $
        \bundlevalue{i}{\bundle{k}} \ge \bundlevalue{i}{\bundle{\ell}}
    $
and
    $
        \bundlevalue{i}{\bundle{\ell}} \ge \bundlevalue{i}{\bundle{k}},
    $
i.e.,
    $
        \bundlevalue{i}{\bundle{k}} = \bundlevalue{i}{\bundle{\ell}},
    $
for all $i,k,\ell \in [n]$.
Since valuations are additive, it follows that in a symmetrically envy free allocation, every bundle is worth exactly $1/n$ to every agent.
This coincides with imposing the symmetric version of \emph{equitability}~\citepar{BraJonKla06_better-ways-to-cut-cake,DubSpa61_Ham-Sandwich},
in which all agents should receive the same value from their assigned bundle (and symmetry would mean that all bundles have the same value for all agents).

An envy free allocation may not exist for indivisible goods.
The primary relaxation we consider is \emph{envy freeness up to one good} (EF1)~\citepar{LipMar04_EF1-implicit,Bud10_EF1-Explicit,Bud11-ef1-approximate-ceei}.
For convenience, define the maximum-valued item in a bundle for agent $i$ as
	\begin{equation*}
        \maxItemValAgentBundle{i}{A} \defeq \max \{ \itemvalue{i}{j} \suchthat j \in A \}.
    \end{equation*}
An agent $i$ is said to be \emph{EF1-satisfied (with their assigned bundle $\bundle{i}$ in an allocation)} if the agent weakly prefers bundle $\bundle{i}$ over any other bundle $\bundle{\ell}$, $\ell \in [n]$, after removing their favorite item from bundle $\bundle{\ell}$; i.e.,
      $\bundlevalue{i}{\bundle{i}} \ge \bundlevalue{i}{\bundle{\ell}} - \maxItemValAgentBundle{i}{\bundle{\ell}}.$
An allocation is said to be EF1 if every agent is EF1-satisfied with their assigned bundle.
EF1 allocations always exist and can be computed efficiently.
The first goal of the project is investigating the corresponding symmetric variant, defined formally below.

\begin{definition}[Symmetrically Envy-Free up to One Good]
    A partition $(\bundle{1},\ldots,\bundle{n})$ is \emph{symmetrically envy free up to one good}, or \emph{symEF1}, when
    every agent $i \in [n]$ is EF1-satisfied when assigned any bundle $\bundle{k},\, k \in [n]$;
    i.e., for every $i,k,\ell \in [n]$, it holds that
    \begin{equation*}
      \bundlevalue{i}{\bundle{k}} \ge \bundlevalue{i}{\bundle{\ell}} - \maxItemValAgentBundle{i}{\bundle{\ell}}.
    \end{equation*}
\end{definition}

\subsection{Related Literature}
\label{sec:related-literature}

This paper adds to a rich literature on the fair allocation of indivisible resources~\citepar{Stei49_first-fair-paper},
building on the notions of envy freeness introduced for \emph{divisible} fair division (``cake cutting'')~\citepar{GamSte58_envy-freeness,Fol67_resource-allocation-EF,Var74_Envy-Free}
and later relaxations to the indivisible setting~\citepar{LipMar04_EF1-implicit,Bud10_EF1-Explicit}.

General symmetric fairness guarantees are attractive theoretically but more challenging to obtain than their nonsymmetric variants,
and may seem unattainable even in the simplest cases,
such as when there are only two agents.
However, optimism can be drawn from the cake cutting context.
An early result to this end is the ``Ham Sandwich Theorem'':
a single hyperplane exists that simultaneously bisects a set of $d$ probability measures in $\R^d$~\citepar{Alon22_fair-partitions-survey}.
This implies that a cake can be split into two halves, such that both parts are valued exactly the same by all (even more than two) agents,
i.e., a so-called \emph{consensus halving} exists in which every agent would be envy free receiving either part~\citepar{Alon87_Discrete-BorUla,SimSu03_Consensus-halving-BorUla}.
As this guarantee holds for more than two agents,
it has immediate implications for \emph{group fairness},
in which the two halves of the cake represent resources that will be shared by a subset of the agents.
Generalizing to $k$ groups, \citeaut{SimSu03_Consensus-halving-BorUla}, then later \citeaut{GoldHoll22_consensus-k-splitting} and \citeaut{ManSukS22_general-symEFk}, find positive results for \emph{consensus $k$-splitting} and the more specific \emph{consensus $1/k$-division},
in which every agent values each of the $k$ bundles equally.
This can be cast as a symmetric fairness problem with respect to either \emph{envy freeness} or \emph{equitability}~\citepar{BraJonKla06_better-ways-to-cut-cake,DubSpa61_Ham-Sandwich}.
Beyond existence, for symmetry in cake cutting settings, the central research questions involve bounding the number of cuts required.

For indivisible goods, \citeaut{ManSukS22_general-symEFk} study \emph{consensus $1/k$-division up to $c$ goods},
where items must be partitioned into $k$ bundles such that each agent is \emph{envy free up to $c$ goods} with any bundle.
Using discrepancy theory, the authors provide bounds on the value of $c$ such that a consensus $1/k$-division up to $c$ goods is guaranteed to exist.
The definition of symEF1 is equivalent when $c=1$ and $k$ equals the number of agents,
but \citeaut{ManSukS22_general-symEFk} do not analyze this particular case in depth.

\section{Special Cases: Existence Under Identical, Disjoint, and Binary Valuations}
\label{sec:existence-examples}

Symmetry adds a stringent constraint to the space of feasible allocations.
We begin by presenting special cases that suffice for existence of a symEF1 allocation.
These serve to reassure us that somewhat broad classes of instances are amenable to symmetrical fairness.
In addition, the key proof idea---coordinating across round-robin allocations---foreshadows our later technical approach.

A round-robin allocation is obtained by having agents pick in a fixed order for every round until all items are selected.
Whenever it is agent $i$'s turn, they add their highest-valued item that remains unallocated to their bundle, $\bundle{i}$.
It is a basic result that this is an EF1 allocation.

We introduce extra notation to facilitate our discussion of round-robin allocations both in this section and later in the paper.
For $i \in [n]$, let $\sigma^i$ be the ranking of the item set $[m]$ in order from most to least favorite item for agent $i$;
i.e., for $j \in [m]$, $\sigma^i(j)$ returns the $j$-th favorite item for agent $i$.
Formally, $\sigma^i$ satisfies the condition $v_{i}(\sigma^{i}(j)) \ge v_{i}(\sigma^{i}(j+1))$ for all $j \in [m-1]$.
This means that $\sigma^{i}(1)$ is the favorite item for agent $i$ and $\sigma^{i}(m)$ is the least favorite item for agent $i$.
For convenience, define $\sigma^{i}(j) \defeq \emptyset$ if $j > m$.

We next define an \emph{agent's round-robin allocation},
obtained by simulating a round-robin procedure for an instance in which all agents are identical.
Bundle $\bundle{\ell}^{i}$ will contain the $\ell$-th item picked by the agent in each of the $\ceil{m/n}$ rounds of the algorithm.

\begin{definition}[Agent's Round-Robin Allocation]
\label{def:agent-round-robin-alloc}
    For agent $i \in [n]$, the \emph{agent's round-robin allocation} is 
    $\alloc^{i} = (\bundle{1}^{i}, \ldots, \bundle{n}^{i})$,
    where
    \begin{equation*}
        \bundle{\ell}^{i} \defeq
            \left\{
                \sigma^{i}( \ell + n(t-1) )
                \text{ for $t \in \{1,\ldots,\ceil{m/n}\}$}
            \right\}
            =
            \left\{
                \sigma^{i}(\ell), \ldots, \sigma^{i}(\ell + n \cdot \floor{m/n})
            \right\}. %
    \end{equation*}
\end{definition}

\begin{lemma}
\label{lem:agent-round-robin-alloc-EF1}
    Agent $i \in [n]$ is EF1-satisfied with any bundle in the round-robin allocation $\alloc^{i}$.
\end{lemma}
\begin{proof}
    The result follows by construction, as an agent's round-robin allocation is akin to $n$ agents identical to agent $i$ selecting items round robin,
    and each of these agents is EF1-satisfied with their assigned bundle.
\end{proof}

Our first special case is \cref{prop:identical-or-disjoint-valuations}, which states that a symEF1 allocation always exists when agents have identical or disjoint valuations.

\begin{proposition}
\label{prop:identical-or-disjoint-valuations}
    Suppose that agents can be partitioned into groups such that (1) all agents in the same group have identical valuations, and (2) agents from different groups have nonzero valuations on mutually disjoint sets of items.
	Then there exists a symEF1 allocation.
\end{proposition}
\begin{proof}
    For simplicity, we assume that each group is a singleton, as the generalization to a larger number of identical agents is straightforward,
    and that every item is positively valued by some agent.
    For each agent $i \in [n]$, let $\alloc^{i}$ be the agent's round-robin allocation \emph{only on the set of items that the agent has nonzero value}.
    By \cref{lem:agent-round-robin-alloc-EF1}, agent $i$ is EF1-satisfied when receiving any bundle in $\alloc^{i}$.
    Define the complete allocation $\alloc = (\bundle{1},\ldots,\bundle{n})$ such that bundle $\bundle{\ell} = \cup_{i \in [n]} \bundle{\ell}^i$ for each $\ell \in [n]$.
    Note that each item is added to exactly one bundle.
    By the assumption of disjoint valuations, the partition $\alloc$ is symEF1.
\end{proof}

Our next special case tells us that a symEF1 allocation always exists when there are two agents with binary valuations. While this may seem like an overly restrictive condition, a yes or no response on items can be a valuable model for many real world situations.
\begin{proposition}
\label{prop:n=2-binary-valuations}
	For two agents with binary valuations,
	a symEF1 allocation always exists.
\end{proposition}
We omit the proof of \cref{prop:n=2-binary-valuations}, as ultimately we generalize this in \cref{cor:n2-exists}.
We state the result in order to help the reader further develop intuition for general settings under which symEF1 allocations will exist.

\section{Existence is Not Guaranteed for Three Agents}
\label{sec:not-guaranteed}
In contrast to the specific settings considered in the previous section, one may imagine that under less restrictive conditions, symEF1 allocations may not exist.
We confirm this suspicion via the following example, %
which implies that for any $n > 2$ agents, a symEF1 allocation may not exist even when valuations are binary.

\begin{example}
\label{ex:not-guaranteed}
	\cref{tab:3agents} gives the values of three agents for four items.
	As there are four items and three bundles,
	some bundle $A_1$ will have two or more items,
	and some agent $i \in [3]$ will have value $v_i(A_1) \ge 2$.
	Since each agent has nonzero value for exactly three of the four items,
	there will be a bundle $A_2$ with $v_i(A_2) = 0$.
	Hence, agent $i$ is not EF1-satisfied when receiving $A_2$.
	Observe that this example can easily be extended to $n$ agents and $n+1$ items.
\exqed
\end{example}
	
	\begin{table}[htb]
	\caption{Valuations for which a symEF1 allocation does not exist with 3 agents and 4 items.}
	\label{tab:3agents}
	\centering
	\begin{tabular}{lcccccc}
		\toprule
	 			& Item $a$ & Item $b$ & Item $c$ & Item $d$ \\
		\midrule
	 	Agent~\player{1} & 1 & 1 & 1 & 0 \\
		Agent~\player{2} & 1 & 1 & 0 & 1 \\
		Agent~\player{3} & 1 & 0 & 1 & 1 \\
	 	\bottomrule
	\end{tabular}
	\end{table}

\section{Existence of SymEF1 Allocations}
\label{sec:symEF1-sufficiency}
\cref{ex:not-guaranteed} in \cref{sec:not-guaranteed} shows that when there are more than two agents and a nontrivial number of items ($m > n$), a symEF1 allocation may not exist.
However, we know from \cref{sec:existence-examples} that rather general families of instances can overcome this barrier.
In this section, we identify a broader sufficient condition for the existence of a symEF1 allocation.
Further, we prove that this sufficient condition is always satisfied for two agents.
We then discuss the uniqueness of symEF1 allocations:
we prove that there are always at least \emph{two} distinct symEF1 allocations
for the two-agent four-item case and give a sufficient condition for \emph{exponentially-many} symEF1 allocations in general.

Our proof extends the intuition of \cref{lem:agent-round-robin-alloc-EF1} that an agent is EF1-satisfied with any bundle in their round-robin allocation.
The challenge is finding a single allocation that aligns all agents' round-robin allocations simultaneously.
We first group the set of items for each agent into so-called ``$n$-tuples''.
We show that an agent is EF1-satisfied with any bundle in an allocation that ``separates'' their $n$-tuples.
Finally, we construct an ``item graph'' such that an $n$-coloring of the graph separates the $n$-tuples for each agent.

\subsection{Separating Tuples of Items}
\label{sec:separating-tuples}

In this section, we partition the set of $m$ items into $\ceil{m/n}$ ``indexed $n$-tuples'' for each agent.
These tuples are constructed by the ordering of each agent's preferences and represent the set of items the agent would pick in one ``round'' when building the agent's round-robin allocation.
Generalizing \cref{lem:agent-round-robin-alloc-EF1}, we show that a symEF1 allocation can be found by ensuring that no bundle contains any items in the same $n$-tuple (for any agent).

We next formally define the set of indexed $n$-tuples for agent $i$,  $\nTuplesForAgent{i} = (\nTupleElementForAgent{i}{1},\ldots,\nTupleElementForAgent{i}{\ceil{m/n}})$.
In this partition, $\nTupleElementForAgent{i}{1}$ is the first $n$ most-valued items for agent $i$, $\nTupleElementForAgent{i}{2}$ is the next $n$, and so on.

\begin{definition}
    The set of \emph{indexed $n$-tuples} for agent $i \in [n]$ is $\nTuplesForAgent{i} = (\nTupleElementForAgent{i}{t})_{t \in [\ceil{m/n}]}$,
    where 
        \begin{equation*}
            \nTupleElementForAgent{i}{t} \defeq \left\{ 
                \sigma^{i}( \ell + n(t-1) )
                \text{ for } \ell \in \left[ \min\{n, m - (t-1)n\} \right] \}
            \right\},
            \quad\mbox{for $t \in [\ceil{m/n}]$.}
        \end{equation*}
\end{definition}

Our first result of this section shows the relationship between $n$-tuples and a fairness criterion that is at least as strong as the symEF1 condition.
We say that an agent's $n$-tuples are \emph{separated by an allocation} if every item in each $n$-tuple is in a different bundle.

\begin{definition}
An agent $i$'s indexed $n$-tuples are \emph{separated} by allocation $(\bundle{1},\ldots,\bundle{n})$ if, for any items $j_{1}, j_{2} \in [m]$, $j_{1} \ne j_{2}$, belonging to the same $n$-tuple $\nTupleElementForAgent{i}{t} \supseteq \{j_{1},j_{2}\}$ for some $t \in [\ceil{m/n}]$, it holds that if $j_{1} \in \bundle{k}$, $k \in [n]$, then $j_{2} \notin \bundle{k}$.
\end{definition}

An agent $i$'s indexed $n$-tuples are separated by the agent's round-robin allocation $\alloc^{i}$ (see \cref{def:agent-round-robin-alloc}).
However, this allocation may not separate other agents' indexed $n$-tuples.

For convenience in the subsequent results, we define notation for the smallest value of an item in a bundle, as a counterpart to the favorite item value $\maxItemValAgentBundle{i}{\bundle{k}}$:
    \begin{equation*}
        \minItemValAgentBundle{i}{\bundle{k}} \defeq \min\{\itemvalue{i}{j} \suchthat j \in \bundle{k}\}.
    \end{equation*}
    
\begin{lemma}\label{lem:separate-tuple-ndm}
    If an allocation $(\bundle{1}, \ldots, \bundle{n})$ separates $\nTuplesForAgent{i}$ for all $i \in [n]$ and $n$ divides $m$,
    then, for all agents $i \in [n]$ and bundles $k,\ell \in [n]$, %
        \begin{equation*}
            \bundlevalue{i}{\bundle{k}} - \minItemValAgentBundle{i}{[m]} \ge \bundlevalue{i}{\bundle{\ell}} - \maxItemValAgentBundle{i}{\bundle{\ell}}.
        \end{equation*}
\end{lemma}

\begin{proof}
    Let $\alloc^{i} = (\bundle{1}^{i}, \ldots, \bundle{n}^{i})$ be agent $i$'s round-robin allocation from \cref{def:agent-round-robin-alloc}.
    Bundle $\bundle{1}^{i}$ contains agent $i$'s highest-valued item from each $n$-tuple.
    Bundle $\bundle{n}^{i}$ consists of agent $i$'s least-valued item from each $n$-tuple.
    Bundle $\bundle{n}^i$, by construction, is agent $i$'s least-valued bundle in $\alloc^{i}$.
    Nonetheless, agent $i$ is EF1-satisfied when receiving bundle $\bundle{n}^{i}$, by \cref{lem:agent-round-robin-alloc-EF1}.
    
    We slightly strengthen the EF1 guarantee.
    We show that agent $i$'s value for the set of items $\bundle{n}^{i} \setminus \sigma^{i}(m)$
    is higher than agent $i$'s value for $\bundle{1}^{i} \setminus \sigma^{i}(1)$,
    which contains all the highest-valued items from each $n$-tuple aside from agent $i$'s favorite item (belonging to tuple $\nTupleElementForAgent{i}{1}$).
    Observe that $\bundle{n}^{i}$ contains an item from every tuple,
    and any item from a tuple $\nTupleElementForAgent{i}{t}$
    is more valuable to agent $i$ than any item in tuple $\nTupleElementForAgent{i}{t+1}$, by construction of the tuples.
    Hence,
        \begin{equation*}
            \bundlevalue{i}{\bundle{n}^{i}} - \bundlevalue{i}{\sigma^{i}(m)}
            = \bundlevalue{i}{\bundle{n}^{i}} - \minItemValAgentBundle{i}{\bundle{n}^{i}}
            = \sum_{t=1}^{\mathclap{m/n-1}} \bundlevalue{i}{\sigma^{i}(tn)}
            \ge \sum_{t=2}^{m/n} \bundlevalue{i}{\sigma^{i}(tn+1)}
            = \bundlevalue{i}{\bundle{1}^{i}} - \maxItemValAgentBundle{i}{\bundle{1}^{i}}.
        \end{equation*}

    Now consider an arbitrary allocation $(\bundle{1}, \ldots, \bundle{n})$ that separates $\nTuplesForAgent{k}$ for all $k \in [n]$.
    As $n$ divides $m$ by assumption, every bundle $\bundle{\ell}$, $\ell \in [n]$, has exactly one item from each $n$-tuple for agent $i$.
    Furthermore, it holds that $\bundlevalue{i}{\bundle{1}^{i}} \ge \bundlevalue{i}{\bundle{\ell}} \ge \bundlevalue{i}{\bundle{n}^{i}}$ for all $\ell \in [n]$,
    by construction.
    It follows that, if agent $i$ receives bundle $\bundle{\ell}$, then, compared to any other bundle $\bundle{k}$,
    \begin{equation*}
        \begin{split}
        \bundlevalue{i}{\bundle{\ell}} - \minItemValAgentBundle{i}{[m]}
            &= \bundlevalue{i}{\bundle{\ell}} - \minItemValAgentBundle{i}{\bundle{n}^{i}}
            \\
            &\ge \bundlevalue{i}{\bundle{n}^{i}} - \minItemValAgentBundle{i}{\bundle{n}^{i}}
            \\
            &\ge \bundlevalue{i}{\bundle{1}^{i}} - \maxItemValAgentBundle{i}{\bundle{1}^{i}} %
            \\
            &\ge \bundlevalue{i}{\bundle{k}} - \maxItemValAgentBundle{i}{\bundle{1}^{i}}
            \\
            &\ge \bundlevalue{i}{\bundle{k}} - \maxItemValAgentBundle{i}{\bundle{1}^{i}} + \left( \maxItemValAgentBundle{i}{\bundle{1}^{i}} - \maxItemValAgentBundle{i}{\bundle{k}} \right)
            \\
            &= \bundlevalue{i}{\bundle{k}} - \maxItemValAgentBundle{i}{A_{k}}. \qedhere
        \end{split}
    \end{equation*}
\end{proof}

This result is stronger than symEF1 if an agent positively values all items; otherwise, if the least-valued item is worth zero, then the inequality reduces to the symEF1 guarantee.

Next, we show that if $n$ does not divide $m$, then we can add items of zero value and reuse \cref{lem:separate-tuple-ndm} to obtain a symEF1 allocation.

\begin{corollary}\label{cor:sep-n-tup}
    If an allocation $\alloc$ separates $\nTuplesForAgent{i}$ for all $i \in [n]$, then $\alloc$ is a symEF1 allocation.
\end{corollary}
\begin{proof}
    Assume that $\alloc = (\bundle{1}, \ldots, \bundle{n})$ is an allocation that separates $\nTuplesForAgent{i}$ for all $i \in [n]$.
    If $n$ divides $m$, then the claim follows from the stronger condition guaranteed by \cref{lem:separate-tuple-ndm}.
    Otherwise, if $m$ is not divisible by $n$, then consider a modified instance with an additional $r = m \bmod n$ items that have zero value for all agents.
    Let $\tilde{\bundle{k}}$ denote a bundle with the same items as $\bundle{k}$, but with one of the new zero-valued items included whenever $\card{\bundle{k}} < \ceil{m/n}$.
    The allocation $\alloc = (\tilde{\bundle{1}},\ldots,\tilde{\bundle{n}})$ separates the $n$-tuples for all agents in the modified instance.
    By \cref{lem:separate-tuple-ndm}, these augmented bundles constitute an allocation for the modified instance in which, for every agent $i$ and bundle $k$,
    the following inequality holds for any bundle $\ell$:
        \begin{equation*}
            \bundlevalue{i}{\tilde{\bundle{k}}} - \minItemValAgentBundle{i}{[m+r]} \ge \bundlevalue{i}{\tilde{\bundle{\ell}}} - \maxItemValAgentBundle{i}{\tilde{\bundle{\ell}}}.
        \end{equation*}
    As 
        $\bundlevalue{i}{\tilde{\bundle{k}}} = \bundlevalue{i}{{\bundle{k}}}$,
        $\maxItemValAgentBundle{i}{\tilde{\bundle{k}}} = \maxItemValAgentBundle{i}{{\bundle{k}}}$,
        $\minItemValAgentBundle{i}{[m+r]} = 0$
    for all $i,k \in [n]$,
    the claim follows.
\end{proof}

An allocation separating the $n$-tuples for all agents is not only symEF1, but also achieves this while keeping all bundles approximately equal size (differing in cardinality by at most one).
This property is called \emph{balanced}~\citepar{kyropoulou2020almost}.
It is neither the case that all balanced allocations are symEF1, nor that all symEF1 allocations must be balanced.

\subsection{Coloring a Graph to Separate Item Tuples}
\label{sec:coloring-graph}
Let $\nTuples = (\nTuplesForAgent{1},\ldots,\nTuplesForAgent{n})$ denote a set of indexed $n$-tuples for a given set of valuations.
As a function of the $n$-tuples, we define $G(\nTuples)$, called the \emph{item graph}, as the graph with vertex set $[m]$ (one vertex per item) 
and edges between every pair of vertices
whose corresponding items
belong to the same $n$-tuple for at least one agent.
Concretely, the edge set is
    \begin{equation*}
        \left\{
                \{j_{\ell}, j_{k}\} \subseteq [m] \times [m] \suchthat j_{\ell} \ne j_{k} \text{ and } \{j_{\ell}, j_{k}\} \subseteq \nTuplesForAgent{i} \text{ for some agent $i \in [n]$ }
            \right\}.
    \end{equation*}
The edges encode ``conflicts'':
if two items in an $n$-tuple for an agent are allocated to the same bundle,
then the allocation will not separate the $n$-tuples for that agent.

A \emph{$k$-coloring} is a partition of a graph's vertices into $k$ classes that induce independent sets, i.e., two vertices receiving the same color cannot be adjacent.
In \cref{thm:symEF1},
we prove that if the item graph is $n$-colorable, then the resulting partition of the items is symEF1.

Before stating and proving this result,
we offer a small example
to build intuition about both the item graph and the $n$-tuples that define the edge set.
\cref{ex:not-necessary} demonstrates that the condition of separating $\nTuplesForAgent{i}$ for all $i \in [n]$, which is represented by the coloring of our item graph, is not necessary for the existence of a symEF1 allocation.

\begin{example}
\label{ex:not-necessary}
	\cref{fig:stuff-not-nece-graph} shows the valuations of three agents for six items.
    The boxed items belong to $\nTupleElementForAgent{i}{2}$ for each agent $i \in [3]$.
    The corresponding item graph is adjacent; the dashed edges represent items that appear in $\nTupleElementForAgent{i}{2}$ for some $i \in [3]$.
    Since the graph contains a 5-clique, it is not 3- (or even 4-) colorable.
    However, the allocation $(\{a,f\}, \{c,e\}, \{b,d\})$ is symEF1.
\end{example}

\begin{figure}
 \centering
    \begin{minipage}{0.65\textwidth}
        \centering
        \renewcommand{\tabcolsep}{2.5pt}
        \begin{tabular}{lcccccc}
            \toprule
                    & Item $a$ & Item $b$ & Item $c$ & Item $d$ & Item $e$ & Item $f$\\
            \midrule
            Agent~\player{1} & \toOne{1} & \toOne{2} & \toOne{3} & {4} & {5} & {6}\\
            Agent~\player{2} & \toOne{1} & \toOne{2} & {4} & \toOne{3} & {5} & {6} \\
            Agent~\player{3} & \toOne{1} & \toOne{2} & {4} & {5} &\toOne{3} & {6} \\
            \bottomrule
        \end{tabular}
    \end{minipage}%
    \begin{minipage}{0.35\textwidth}
        \centering
         \begin{tikzpicture}[node distance={13mm}, thick, main/.style = {draw, circle, minimum height=5mm, inner sep=0pt, opacity=1.,scale=1}]
            \def\radius{1.25}
        
            \coordinate (a) at ({\radius*cos(234)}, {\radius*sin(234)});
            \coordinate (b) at ({\radius*cos(162)}, {\radius*sin(162)});
            \coordinate (c) at ({\radius*cos(90)}, {\radius*sin(90)});
            \coordinate (d) at ({\radius*cos(18)}, {\radius*sin(18)});
            \coordinate (e) at ({\radius*cos(306)}, {\radius*sin(306)});
            \coordinate (f) at ($(d)+(\radius,0)$);

            \node[main, fill=mywhite] (itema) at (a) {$a$};
            \node[main, fill=black, text=white] (itemb) at (b) {$b$};
            \node[main, fill=mylightblue, pattern=vertical lines, pattern color=mylightblue] (itemc) at (c) {$c$};
            \node[main, fill=myorange] (itemd) at (d) {$d$};
            \node[main, fill=mygreen1, pattern=crosshatch, pattern color=mygreen1] (iteme) at (e) {$e$};
            \node[main, fill=mywhite] (itemf) at (f) {$f$};

            \draw [color=purple, dashed, thick] (itemb)--(itemd);
            \draw [color=purple, dashed, thick] (itema)--(itemb);
            \draw [color=purple, dashed, thick] (itema)--(itemc);
            \draw [color=purple, dashed, thick] (itemb)--(itemc);
            \draw [color=purple, dashed, thick] (itema)--(itemd);
            \draw [color=purple, dashed, thick] (itemb)--(itemd);
            \draw [color=purple, dashed, thick] (itema)--(iteme);
            \draw [color=purple, dashed, thick] (itemb)--(iteme);
            \draw (itemd)--(iteme);
            \draw (itemd)--(itemf);
            \draw (iteme)--(itemf);
            \draw (itemc)--(iteme);
            \draw (itemc)--(itemf);
            \draw (itemc)--(itemd);
        \end{tikzpicture}
    \end{minipage}
	\caption{An instance with three agents and six items in which the item graph is not $3$-colorable, though a symEF1 allocation exists, showing that the sufficient condition in \cref{thm:symEF1} is not necessary. The boxed items represent indexed $n$-tuple $\nTupleElementForAgent{i}{2}$ for each agent $i \in [3]$ and lead to the dashed edges in the item graph. A 5-coloring of the graph is given.}
     \label{fig:stuff-not-nece-graph}
 \end{figure}
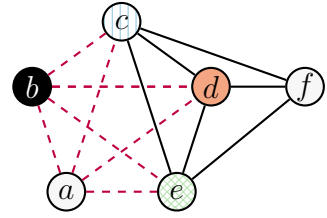

We now state \cref{thm:symEF1}, which proves that $n$-colorability of the item graph is a sufficient condition for existence of a symEF1 allocation.

\begin{theorem}
\label{thm:symEF1}
    If the item graph $G(\nTuples)$ is $n$-colorable, then there exists a symEF1 allocation.
\end{theorem}

\begin{proof}
    Assume that the item graph $G(\nTuples)$ is $n$-colorable.
    For each $\ell \in [n]$, let $\bundle{\ell}$ consist of the set of items assigned color $\ell$.
    For every pair of items $j_1, j_2 \in [m]$ that belong to the same $n$-tuple for any agent,
    there is an edge between the corresponding vertices of $G(\nTuples)$.
    By definition of coloring, these items must receive different colors, and hence belong to different bundles.
    As a result, the partition $(\bundle{1},\ldots,\bundle{n})$ separates the indexed $n$-tuples for all agents,
    and by \cref{cor:sep-n-tup} is symEF1.
\end{proof}

This result allows us to prove that for two agents there always exists a symEF1 allocation.
\begin{corollary}\label{cor:n2-exists}
    For two agents, a symEF1 allocation exists.
\end{corollary}

\begin{proof}
    We will prove that every item graph for a two-agent instance has no odd cycles, so it is bipartite and thus 2-colorable.
    Assume without loss of generality that $m$ is even (e.g., by adding a zero-valued item).
    Let $G[C]$ denote the subgraph of the item graph $G(\nTuples)$ \emph{induced} by a vertex subset $C$, defined by vertex set $C$ and any edges of $G(\nTuples)$ between vertices in $C$.

    Every vertex of $G(\nTuples)$ has degree at most two, since each of the agents contributes an edge incident to that vertex.
    The degree of a vertex is one when it belongs to a $2$-tuple that appears for both agents.
    As the adjacent vertex is in the same $2$-tuple, it also has degree one, creating an isolated edge in the item graph,
    which is an even cycle of length 2.
    The vertices of the item graph that are not in isolated edges all have degree 2;
    hence, the subgraph induced by these vertices is a set of cycles.
    We show that these cycles are all even.

    Let $\{u_0,v_0\}$ be an arbitrary $2$-tuple for agent 1,
    and denote by $C$ the vertices of the connected component of $G(\nTuples)$ containing $u_0$ and $v_0$.
    We proceed by iteratively identifying more vertices in $C$, starting with $C_0 \defeq \{u_0,v_0\}$.
    At step $k$, we will be given a set of $2k$ vertices in $C$, denoted by $C_{k-1}$.
    We will either determine that $C = C_{k-1}$ induces an even cycle $G[C]$,
    or we find two additional vertices of $C$ that are adjacent to those in $C_{k-1}$.
    
    For step 1, there are two possibilities.
    If $\{u_0,v_0\}$ is also a $2$-tuple for agent 2, then $C = C_0 = \{u_0,v_0\}$ is an isolated edge in the item graph, i.e., an even cycle of length 2.
    If not, then $u_0$ and $v_0$ belong to different $2$-tuples for agent 2,
    say $\{u_0,u_1\}$ and $\{v_0,v_1\}$.
    Let $C_1 \defeq C_0 \cup \{u_1,v_1\} = \{u_1,u_0,v_0,v_1\}$.

    Continuing inductively, at the start of step $k$, the set $C_{k-1} = \{u_{k-1},\ldots,u_0,v_0,\ldots,v_{k-1}\}$ has (even) cardinality $2k$.
    We analyze whether we can grow $C_{k-1}$ using the $2$-tuples of agent $i$, where $i=1$ when $k$ is even and $i=2$ when $k$ is odd.
    If $\{u_{k-1},v_{k-1}\}$ is a $2$-tuple for agent $i$, then $C = C_{k-1}$ induces a $2k$-cycle in the item graph.
    Otherwise, $u_{k-1}$ and $v_{k-1}$ belong to different $2$-tuples for agent $i$, which implies that two additional vertices, say $u_{k}$ and $v_{k}$,
    belong to the connected component $C$, constituting $C_{k} \defeq C_{k-1} \cup \{u_{k}, v_{k}\}$.

    Since this process maintains an even cardinality subset of $C$ at each step and eventually terminates due to the finiteness of the vertex set, the cycle $G[C]$ is even.
    A graph with no odd cycles is 2-colorable, which implies that a symEF1 allocation exists by \cref{thm:symEF1}.
\end{proof}

An alternate proof of \cref{cor:n2-exists} can be derived from results on matroids by \citeaut{edmonds1970submodular,edmonds2003submodular}, or equivalently through the bihierarchy framework of \citeaut{BudChe13_BiHi}, by showing that an allocation that can separate $\nTuplesForAgent{1}$ and $\nTuplesForAgent{2}$ is obtainable via an integer program with constraints defined by two laminar families (one for each of the agents).
Specifically, the resulting constraint matrix can be proved to be totally unimodular, implying the existence of an integer solution corresponding to the desired allocation,
and hence a symEF1 allocation exists for two agents via \cref{cor:sep-n-tup}.
Notably, total unimodularity is not guaranteed for more agents; \cref{thm:symEF1} can be viewed as a generalization of this technique.

While \cref{thm:symEF1} provides a sufficient condition for the existence of a symEF1 allocation,
we know from \cref{ex:not-guaranteed} that, for three or more agents, symEF1 allocations do not always exist.
However, as our last positive result for this line of inquiry, \cref{prop:n-tuples} implies that when the sufficient condition is met for one instance,
then existence is assured for a whole family of related instances that maintain the sets of indexed $n$-tuples for all agents.

\begin{proposition}\label{prop:n-tuples}
    Let $G(\nTuples)$ be an $n$-colorable item graph associated to a particular instance of $n$ agents' valuations for $m$ items.
    Let $ \overline{\nTuples} = \{ \overline{\nTuplesForAgent{1}}, \ldots, \overline{\nTuplesForAgent{n}} \} $ be the set of indexed $n$-tuples for a different instance with $n$ agents and $m$ items.
    If there exists a permutation $\sigma$ of the agents and a permutation $\pi^i$ of the indexed $n$-tuples for each agent $i \in [n]$ such that $\pi^i(\overline{\nTuplesForAgent{\sigma(i)}}) = {\nTuplesForAgent{i}}$,
    then $G(\overline{\nTuples}) = G(\nTuples)$.
\end{proposition}

\begin{proof}
    The claim is immediate as the vertices and edges are the same for $G(\overline{\nTuples})$ and $G(\nTuples)$.
\end{proof}

In particular, the above result means a symEF1 allocation obtained via the sufficient condition from \cref{thm:symEF1} remains symEF1 if one or more agents completely reverse their valuations or arbitrarily reorder their preferences within every $n$-tuple.

\subsection{Lower Bounding the Number of SymEF1 Allocations}
\label{sec:lower-bound-alloc}

When a symEF1 allocation exists for an instance, the next natural question is whether it is unique.
Our graph construction and \cref{thm:symEF1} leads to a simple, constructive answer.
We show how to find \emph{multiple} symEF1 allocations and indeed \emph{exponentially many} in the number of components of the item graph.

We say that two allocations of the item set $[m]$ are \emph{distinct} if the partitions they imply are not identical, i.e., there is some bundle in one allocation that is not identical to a bundle in the second allocation.
To avoid ``trivial'' cases in which distinct allocations occur because items are interchangeable,
we assume that items are \emph{distinct}:
each item has positive value for at least one agent,
and no two items are identical.

\begin{theorem}\label{thm:lowerbound-symef1}
    Given a set of indexed $n$-tuples $\nTuples$ for an instance with distinct items,
    let $C$ denote the number of components in the item graph $G(\nTuples)$.
    If $G(\nTuples)$ is $n$-colorable,
    then the number of distinct symEF1 allocations is at least $(n!)^{C-1}$.
\end{theorem}
\begin{proof}
    Assume that the item graph $G(\nTuples)$ has $C$ connected components and is $n$-colorable. From \cref{thm:symEF1} we know that an $n$-coloring of the item graph corresponds to a symEF1 allocation. We wish to only count distinct allocations. If we find an initial coloring of the graph and fix the coloring of a single component,
    then any new $n$-coloring on the other $C-1$ components returns a distinct symEF1 allocation.
    
    Since the cardinality of a tuple is $n$ we know that a given component has a minimum of $n$ vertices. Additionally, since we assumed the graph to be $n$-colorable, and there are at least $n$ vertices, we can simply permute these colors to obtain $n!$ distinct symEF1 allocations. Since there are $C-1$ components that are not fixed, there are at least $(n!)^{C-1}$ distinct symEF1 allocations.
\end{proof}

While the construction of this bound is simple,
it is useful for many instances of the problem.
For example, suppose all agents are identical, or more generally, in each agent's round-robin allocation, the same set of items is picked in each round.
Then all the agents have identical indexed $n$-tuples.
As a result, the item graph has $\ceil{m/n}$ components.
Following this reasoning, when the set of indexed $n$-tuples is nearly identical across agents, the resulting item graph may have more connected components, leading to a multitude of distinct symEF1 allocations.
Depending on the resources being allocated, you might expect a high correlation across agent preferences, 
which would mean tuples that are similar and thus graphs with more components.

From \cref{thm:lowerbound-symef1}, we have a lower bound on the number of distinct symEF1 allocations for a problem instance that satisfies \cref{thm:symEF1}. Enumerating the number of distinct $n$-colorings of a graph is difficult. Therefore, in general this lower bound is not tight. In the next section, we study a small instance of the problem to gain better intuition about this lower bound.

\subsection{Two SymEF1 Allocations for Two Agents and Four Items}
\label{sec:2agents-4items}

We have proved that a symEF1 allocation always exists for two agents, but one might expect that the symmetry restriction would make such allocations rare or even unique for some instances.
This is indeed the case for two agents and three items, as we see in \cref{ex:unique-symef1-n2-m3}.

\begin{example}
\label{ex:unique-symef1-n2-m3}
    Consider the valuations of two agents for three items shown below.

\begin{center}
    \begin{tabular}{l *{3}{c}}
        \toprule
                & Item $a$ & Item $b$ & Item $c$ \\
        \midrule
        Agent~\player{1} & {$1$} & {$1/2$\phantom{$\ +\epsilon$}} & {$1/2+\epsilon$} \\
        Agent~\player{2} & {$1$} & {$1/2 + \epsilon$} & {$1/2$\phantom{$\ +\epsilon$}} \\
        \bottomrule
    \end{tabular}
\end{center}
    The bundle containing item $a$ cannot have item $b$ or $c$, as the agent that values that item at $1/2+\epsilon$ would not be EF1-satisfied compared to the other bundle that would have value $1/2$.
    Hence, the only distinct symEF1 partition is $(\{a\},\{b,c\})$.
    \exqed
\end{example}

We next analyze instances with two agents and four items,
for which we prove that, perhaps surprisingly, there always exist at least two distinct symEF1 allocations.

\begin{proposition}\label{prop-no-unique-2}
    Any instance with two agents and four distinct items admits at least two distinct symEF1 allocations.
\end{proposition}

We defer the somewhat tedious case-based proof to \Cref{app:proofs}.
We instead show why the item graph construction and vertex coloring approach of \cref{thm:symEF1} and \cref{cor:n2-exists} is insufficient to prove this result, while also offering some insight into this class of instances.

We label the four items $a$, $b$, $c$, and $d$, where (without loss of generality) $a$ and $b$ refer to elements of the first $2$-tuple for agent 1, $\nTupleElementForAgent{1}{1}$, and $c$ and $d$ refer to the items in $\nTupleElementForAgent{1}{2}$.
If the $2$-tuples for agent 2 also pair items this way, then the item graph is two isolated edges (equivalently, two disconnected 2-cycles, allowing for parallel edges), shown in \cref{disconnected-cycle}.
In this case, there are two valid $2$-colorings leading to distinct partitions $(\{a,c\},\{b,d\})$ and $(\{a,d\},\{b,c\})$.

Otherwise, again without loss of generality, assume that agent 2 pairs item $a$ with item $c$ in one of their $2$-tuples.
This implies that the item graph consists of a single $4$-cycle, as seen in \cref{single-cycle}.
This graph admits a unique coloring (up to isomorphism),
which would suggest a unique symEF1 partition,
but \cref{prop-no-unique-2} proves the existence of a second distinct symEF1 allocation.

This result can be extended. For any number of items $m \ge 4$ and two agents,
suppose agent 2's valuations are obtained from agent 1's by making agent 2's least-favorite item equal to agent 1's most-valued item, i.e., $\sigma^2(j) = \sigma^1(j+1 \bmod m)$ for $j \in [m]$.
The resulting item graph is an even cycle having a unique 2-coloring, but it is easy to find such instances with nonunique symEF1 allocations.
Indeed, we conjecture that this holds for all instances with $n$ agents and $m > n$ items; see \Cref{con:2-symef1-exists}.

\begin{figure}
    \centering
    \begin{subfigure}[t]{0.45\textwidth}
        \centering
        \begin{tikzpicture}[node distance={20mm}, thick, main/.style = {draw, circle}]
            \node[main] (1) {$a$};
            \node[main] (2) [right of=1] {$b$};
            \node[main] (3) [below of=1] {$c$}; 
            \node[main] (4) [right of=3] {$d$};
            \draw (1) to [out=45, in=135, looseness=1] (2);
            \draw (1) to [out=315, in=225, looseness=1] (2);
            \draw (3) to [out=45, in=135, looseness=1] (4);
            \draw (3) to [out=315, in=225, looseness=1] (4);
        \end{tikzpicture}
        \caption{A graph consisting of two disconnected even cycles has two distinct colorings leading to two unique symEF1 allocations via the sufficient condition of \cref{thm:symEF1}.}
        \label{disconnected-cycle}
    \end{subfigure}\hfill
    \begin{subfigure}[t]{0.45\textwidth}
        \centering
        \begin{tikzpicture}[node distance={20mm}, thick, main/.style = {draw, circle}]
            \node[main] (1) {$a$};
            \node[main] (2) [right of=1] {$b$};
            \node[main] (3) [below of=1] {$c$}; 
            \node[main] (4) [right of=3] {$d$};
            \draw (1)--(2);
            \draw (2)--(4);
            \draw (4)--(3);
            \draw (3)--(1);
        \end{tikzpicture}
        \caption{A single cycle graph which has a unique coloring and thus has a unique symEF1 allocation by our sufficient condition. We know for $n=2$ and $m=4$ there exists at least two symEF1 allocations.\label{single-cycle}}
    \end{subfigure}
\end{figure}

\section{Relationship with Existing Fairness Criteria}
\label{sec:comparision-with-existing-criteria}

In this section, we establish that symEF1 does not coincide with other strengthenings of EF1 that, on the surface, appear to be closely related to the property of symEF1 allocations that all bundles are nearly equal in value.

\subsection{No Guarantees for Symmetric Envy Freeness up to Any Good}

Recently, there has been concerted effort to determine if there always exists an allocation that is \emph{envy free up to \emph{any} good} (\emph{EFX})~\citepar{gourves2014near,CarKur19_EF1-MNW-PO},
in which an agent's envy is bounded by the \emph{minimum} value of an item in any envied bundle, rather than the weaker maximum in an EF1 allocation.
Existence of an EFX allocation has been shown for three agents and some additional special cases~\citepar{AmaBirFilHolVou21_MNW-EFX,chaudhury2024EFX}, but in general it remains an open question.

An EFX allocation has similarities to a symEF1 partition in that an agent can only value their assigned bundle a little less than any other bundle.
Special cases of existence for EFX allocations have been studied that also guarantee symEF1 allocations.
For example,
    \citeaut{PlaRou21_EFX-same-ordering} %
prove that if all agents rank items in the same order, then an EFX allocation exists;
\cref{thm:symEF1} implies that a symEF1 allocation exists for this setting, as the set of indexed $n$-tuples is the same for all agents.

However, in general, the two concepts of symEF1 and EFX do not imply one another.
While an EFX allocation requires that an agent's assigned bundle has high value to the agent relative to other bundles, it allows for situations in which the agent has little value for these other bundles;
this would prevent the allocation from being symmetrically fair.
\cref{ex:not-guaranteed} contains an instance with no symEF1 allocation, but an EFX allocation exists as there only three agents.
Conversely, a symEF1 allocation might not be EFX, since removing the maximum-valued item (to satisfy EF1) permits significantly more flexibility.
For instance, if two agents identically value three items at $1$, $1/2$, and $1/2$,
then the allocation putting a $1/2$-valued item into its own bundle is symEF1 but not EFX.

We now briefly consider a symmetric extension of EFX. 

\begin{definition}
An allocation is \emph{symmetrically envy free up to any good} (\emph{symEFX}) if
    \begin{equation*}
        \bundlevalue{i}{\bundle{k}} \ge v_{i}(A_{\ell}) - \minItemValAgentBundle{i}{A_{\ell}}, \quad\text{for all } i,k,\ell \in [n].
    \end{equation*}
\end{definition}
The following example shows that a symEFX allocation is not guaranteed to exist for any number of agents, even when a symEF1 allocation does exist.

\begin{example}
    Given $n$ agents and $m = n + 1$ items, for $i \in [n], j \in [m]$, define valuations
    \begin{equation*}
        \itemvalue{i}{j} = 
        \begin{cases}
            1 & \text{if } i=j, \\
            \varepsilon & \text{if } i \neq j.
        \end{cases}
    \end{equation*}
    In any allocation, some bundle will have at least two items, and there is an agent that will value this bundle at least $1+\epsilon$ and every other bundle at most $\epsilon$.
    Hence, a symEFX allocation does not exist.
    However, a symEF1 allocation exists, by placing each item $j \in [n]$ in a separate bundle and item $n+1$ in any bundle.
    \exqed
 \end{example}

\subsection{Maximum Nash Welfare Solutions May Not Be SymEF1}
\label{sec:max-nash-welfare}

Another focus in the literature has been on allocations that maximize the \emph{Nash welfare}, defined as the product of the agents' values for their assigned bundles.
Such an allocation is referred to as a \emph{maximum Nash welfare} (\emph{MNW}) solution, and is known to provide an allocation that is both EF1 and \emph{Pareto optimal} under additive valuations~\citepar{CarKurMouProShaWan16_ec-unreasonable-fairness-mnw,CarKur19_EF1-MNW-PO}.
An allocation is Pareto optimal if, in any other allocation in which some agent has strictly higher value for their assigned bundle, some agent has strictly lower value.

One may wonder whether there is a relation between an MNW solution and a symEF1 allocation.
The following example shows how the two concepts may differ.

\begin{example}
\label{ex:MNW}

\cref{tab:MNW} gives the values of two agents for six items.
\cref{fig:MNWex-MNWsol} shows an MNW allocation on the left, which is not symEF1, and a symEF1 allocation on the right. \exqed
\end{example}

\begin{table}
\centering
\caption{Valuations for two agents and six items such that the MNW solution is not a symEF1 allocation.}
\label{tab:MNW}
\begin{tabular}{lcccccc}
	\toprule
	 			& Item $a$ & Item $b$ & Item $c$ & Item $d$ & Item $e$ & Item $f$ \\
	\midrule
	 Agent~\player{1} & 1 & 2 & 3 & 4 & 5 & 6 \\
	 Agent~\player{2} & 3 & $1+\eps$ & 3 & 1 & 3 & $1+2\eps$ \\
	 \bottomrule
\end{tabular}
\end{table}

We note that because an MNW solution is guaranteed to be EF1, for any set of valuations such that the MNW solution is not a symEF1 allocation, we incur a \emph{social welfare price} to choosing a symEF1 allocation, i.e., the total value across all agents decreases.

\begin{figure}
\centering
\begin{subfigure}[t]{0.45\textwidth}
\centering
  \begin{tikzpicture}[scale=0.45]
    \coordinate (start) at (0,0);
    \coordinate (item_width) at (1,0);
	\coordinate (bundle1_translate) at (0,0);
	\coordinate (bundle2_translate) at ($2*(item_width)$);
    \coordinate (agent1_translate) at (0,0);
    \coordinate (agent2_translate) at ($2*(bundle2_translate)+(item_width)$);
    \coordinate (interdistance) at ($2*(item_width)$);
	\coordinate (arrow_start) at ($(agent2_translate)+(bundle2_translate)+(item_width)+(interdistance)$);
	\coordinate (arrow_length) at ($3*(item_width)$);
    \coordinate (translate) at ($(arrow_start)+(arrow_length)+(interdistance)$);
	\coordinate (translate2) at ($2*(translate)$);
	\coordinate (brace_height) at (0,13.5);
	
    \coordinate (p1_1_ht) at (0,1);
    \coordinate (p1_2_ht) at (0,2);
    \coordinate (p1_3_ht) at (0,3);
    \coordinate (p1_4_ht) at (0,4);
    \coordinate (p1_5_ht) at (0,5);
    \coordinate (p1_6_ht) at (0,6);
    \coordinate (p2_1_ht) at (0,3);
    \coordinate (p2_2_ht) at (0,1.0001);
    \coordinate (p2_3_ht) at (0,3);
    \coordinate (p2_4_ht) at (0,1);
    \coordinate (p2_5_ht) at (0,3);
    \coordinate (p2_6_ht) at (0,1.0002);
    
  	\def\BundleOneItemOne{$b$}
	\def\BundleOneItemTwo{$d$}
	\def\BundleOneItemThree{$f$}
	
    \coordinate (p1_11_ht) at (p1_2_ht);
    \coordinate (p1_12_ht) at (p1_4_ht);
    \coordinate (p1_13_ht) at (p1_6_ht);
    
    \coordinate (p2_11_ht) at (p2_2_ht);
    \coordinate (p2_12_ht) at (p2_4_ht);
    \coordinate (p2_13_ht) at (p2_6_ht);
    
  	\def\BundleTwoItemOne{$a$}
	\def\BundleTwoItemTwo{$c$}
	\def\BundleTwoItemThree{$e$}
	
    \coordinate (p1_21_ht) at (p1_1_ht);
    \coordinate (p1_22_ht) at (p1_3_ht);
    \coordinate (p1_23_ht) at (p1_5_ht);
    
    \coordinate (p2_21_ht) at (p2_1_ht);
    \coordinate (p2_22_ht) at (p2_3_ht);
    \coordinate (p2_23_ht) at (p2_5_ht);
    
	\coordinate (p1) at ($(start)+1.5*(item_width)+(0,-0.75)$);
    \coordinate (p1_11) at ($(start)+(agent1_translate)$);
    \coordinate (p1_12) at ($(p1_11_ht)+(agent1_translate)$);
    \coordinate (p1_13) at ($(p1_11_ht)+(p1_12_ht)+(agent1_translate)$);
    \coordinate (p1_21) at ($(start)+(bundle2_translate)+(agent1_translate)$);
    \coordinate (p1_22) at ($(p1_21_ht)+(bundle2_translate)+(agent1_translate)$);
    \coordinate (p1_23) at ($(p1_21_ht)+(p1_22_ht)+(bundle2_translate)+(agent1_translate)$);
    
	\coordinate (p2) at ($(start)+1.5*(item_width)+(0,-0.75)+(agent2_translate)$);
    \coordinate (p2_11) at ($(start)+(agent2_translate)$);
    \coordinate (p2_12) at ($(p2_11_ht)+(agent2_translate)$);
    \coordinate (p2_13) at ($(p2_11_ht)+(p2_12_ht)+(agent2_translate)$);
    \coordinate (p2_21) at ($(start)+(bundle2_translate)+(agent2_translate)$);
    \coordinate (p2_22) at ($(p2_21_ht)+(bundle2_translate)+(agent2_translate)$);
    \coordinate (p2_23) at ($(p2_21_ht)+(p2_22_ht)+(bundle2_translate)+(agent2_translate)$);
    
    \node at (p1) () {Agent 1};
    \draw[draw=black,fill=white] ($(p1_11)$) rectangle ++($(p1_11_ht)+(item_width)$) node [pos=0.5] {\BundleOneItemOne};
    \draw[draw=black,fill=white] ($(p1_12)$) rectangle ++($(p1_12_ht)+(item_width)$) node [pos=0.5] {\BundleOneItemTwo};
    \draw[draw=black,fill=white] ($(p1_13)$) rectangle ++($(p1_13_ht)+(item_width)$) node [pos=0.5] {\BundleOneItemThree};
    \draw[draw=black,fill=white] (p1_21) rectangle ++($(p1_21_ht)+(item_width)$) node [pos=0.5] {\BundleTwoItemOne};
    \draw[draw=black,fill=white] (p1_22) rectangle ++($(p1_22_ht)+(item_width)$) node [pos=0.5] {\BundleTwoItemTwo};
    \draw[draw=black,fill=white] (p1_23) rectangle ++($(p1_23_ht)+(item_width)$) node [pos=0.5] {\BundleTwoItemThree};

    \node at (p2) () {Agent 2};
    \draw[draw=black,fill=white] ($(p2_11)$) rectangle ++($(p2_11_ht)+(item_width)$) node [pos=0.5] {\BundleOneItemOne};
    \draw[draw=black,fill=white] ($(p2_12)$) rectangle ++($(p2_12_ht)+(item_width)$) node [pos=0.5] {\BundleOneItemTwo};
    \draw[draw=black,fill=white] ($(p2_13)$) rectangle ++($(p2_13_ht)+(item_width)$) node [pos=0.5] {\BundleOneItemThree};
    \draw[draw=black,fill=white] (p2_21) rectangle ++($(p2_21_ht)+(item_width)$) node [pos=0.5] {\BundleTwoItemOne};
    \draw[draw=black,fill=white] (p2_22) rectangle ++($(p2_22_ht)+(item_width)$) node [pos=0.5] {\BundleTwoItemTwo};
    \draw[draw=black,fill=white] (p2_23) rectangle ++($(p2_23_ht)+(item_width)$) node [pos=0.5] {\BundleTwoItemThree};

	\end{tikzpicture}
	\caption{An MNW allocation with $A_1 = \{b,d,f\}$ and $A_2 = \{a,c,e\}$.
	The Nash welfare of this allocation is $v_1(A_1) \cdot v_2(A_2) = 12 \cdot 9 = 108$.}
	\label{fig:MNWalloc}
  \end{subfigure}
  \quad
  \begin{subfigure}[t]{0.45\textwidth}
  \centering
    \begin{tikzpicture}[scale=0.5]
    
    \coordinate (start) at ($(start)+(translate)$);

  	\def\BundleOneItemOne{$c$}
	\def\BundleOneItemTwo{$d$}
	\def\BundleOneItemThree{$f$}

    \coordinate (p1_11_ht) at (p1_3_ht);
    \coordinate (p1_12_ht) at (p1_4_ht);
    \coordinate (p1_13_ht) at (p1_6_ht);
    
    \coordinate (p2_11_ht) at (p2_3_ht);
    \coordinate (p2_12_ht) at (p2_4_ht);
    \coordinate (p2_13_ht) at (p2_6_ht);    
  	\def\BundleTwoItemOne{$a$}
	\def\BundleTwoItemTwo{$b$}
	\def\BundleTwoItemThree{$e$}
    
    \coordinate (p1_21_ht) at (p1_1_ht);
    \coordinate (p1_22_ht) at (p1_2_ht);
    \coordinate (p1_23_ht) at (p1_5_ht);
    
    \coordinate (p2_21_ht) at (p2_1_ht);
    \coordinate (p2_22_ht) at (p2_2_ht);
    \coordinate (p2_23_ht) at (p2_5_ht);

	\coordinate (p1) at ($(start)+1.5*(item_width)+(0,-0.75)$);
    \coordinate (p1_11) at ($(start)+(agent1_translate)$);
    \coordinate (p1_12) at ($(start)+(p1_11_ht)+(agent1_translate)$);
    \coordinate (p1_13) at ($(start)+(p1_11_ht)+(p1_12_ht)+(agent1_translate)$);
    \coordinate (p1_21) at ($(start)+(bundle2_translate)+(agent1_translate)$);
    \coordinate (p1_22) at ($(start)+(p1_21_ht)+(bundle2_translate)+(agent1_translate)$);
    \coordinate (p1_23) at ($(start)+(p1_21_ht)+(p1_22_ht)+(bundle2_translate)+(agent1_translate)$);
    
	\coordinate (p2) at ($(start)+1.5*(item_width)+(0,-0.75)+(agent2_translate)$);
    \coordinate (p2_11) at ($(start)+(agent2_translate)$);
    \coordinate (p2_12) at ($(start)+(p2_11_ht)+(agent2_translate)$);
    \coordinate (p2_13) at ($(start)+(p2_11_ht)+(p2_12_ht)+(agent2_translate)$);
    \coordinate (p2_21) at ($(start)+(bundle2_translate)+(agent2_translate)$);
    \coordinate (p2_22) at ($(start)+(p2_21_ht)+(bundle2_translate)+(agent2_translate)$);
    \coordinate (p2_23) at ($(start)+(p2_21_ht)+(p2_22_ht)+(bundle2_translate)+(agent2_translate)$);
    
    \node at (p1) () {Agent 1};
    \draw[draw=black,fill=white] ($(p1_11)$) rectangle ++($(p1_11_ht)+(item_width)$) node [pos=0.5] {\BundleOneItemOne};
    \draw[draw=black,fill=white] ($(p1_12)$) rectangle ++($(p1_12_ht)+(item_width)$) node [pos=0.5] {\BundleOneItemTwo};
    \draw[draw=black,fill=white] ($(p1_13)$) rectangle ++($(p1_13_ht)+(item_width)$) node [pos=0.5] {\BundleOneItemThree};
    \draw[draw=black,fill=white] (p1_21) rectangle ++($(p1_21_ht)+(item_width)$) node [pos=0.5] {\BundleTwoItemOne};
    \draw[draw=black,fill=white] (p1_22) rectangle ++($(p1_22_ht)+(item_width)$) node [pos=0.5] {\BundleTwoItemTwo};
    \draw[draw=black,fill=white] (p1_23) rectangle ++($(p1_23_ht)+(item_width)$) node [pos=0.5] {\BundleTwoItemThree};

    \node at (p2) () {Agent 2};
    \draw[draw=black,fill=white] ($(p2_11)$) rectangle ++($(p2_11_ht)+(item_width)$) node [pos=0.5] {\BundleOneItemOne};
    \draw[draw=black,fill=white] ($(p2_12)$) rectangle ++($(p2_12_ht)+(item_width)$) node [pos=0.5] {\BundleOneItemTwo};
    \draw[draw=black,fill=white] ($(p2_13)$) rectangle ++($(p2_13_ht)+(item_width)$) node [pos=0.5] {\BundleOneItemThree};
    \draw[draw=black,fill=white] (p2_21) rectangle ++($(p2_21_ht)+(item_width)$) node [pos=0.5] {\BundleTwoItemOne};
    \draw[draw=black,fill=white] (p2_22) rectangle ++($(p2_22_ht)+(item_width)$) node [pos=0.5] {\BundleTwoItemTwo};
    \draw[draw=black,fill=white] (p2_23) rectangle ++($(p2_23_ht)+(item_width)$) node [pos=0.5] {\BundleTwoItemThree};

  \end{tikzpicture}
  \caption{A symEF1 allocation with $A_1 = \{c,d,f\}$ and $A_2 = \{a,b,e\}$, which has a Nash welfare of $v_1(A_1) \cdot v_2(A_2) = 13 \cdot 7 = 91$.}
  \label{fig:symEF1alloc1}
  \end{subfigure}
  \caption{MNW and symEF1 allocations for ~\cref{ex:MNW}. The MNW allocation is EF1, and in fact entirely envy free, with agent $1$ receiving bundle $1$, but \emph{not} symEF1.}
  \label{fig:MNWex-MNWsol}
\end{figure}

Our initial exploration of extensions of symmetric fairness beyond EF1 is limited and leaves many settings to explore.
A promising direction is imposing additional assumptions on the valuation functions, such as seeking guarantees when there are only two (or a small constant number) distinct valuation functions (agent types), as studied by \citeaut{bu2023fair}.

\section{Simulation Results with Five or Fewer Agents}
\label{sec:simulation}

While we have proved that symEF1 allocations exist under certain restrictive assumptions, these conditions need not necessarily hold for instances to admit a symEF1 partition.
In this section, we report simulation results for $n \le 5$ agents to approximate the density of symEF1 allocations while varying both the number of items $m$ and the valuations for the set of agents.
We provide an integer programming formulation for obtaining a symEF1 allocation.
To reduce optimization-related bottlenecks in running time, we also introduce a greedy heuristic that is often effective at finding symEF1 allocations.
The heuristic relies on progressively more complex local updates,
and we analyze how often these are required,
which may be of independent interest in seeking instances that humans find challenging.
Finally, we evaluate if the range of valuations plays a role in the existence of symEF1 allocations.

\subsection{Verifying Existence of SymEF1 Allocations}
\label{sec:simulation:mip}

Since our target is exactly identifying which instances admit symEF1 allocations,
and we have shown in \cref{ex:not-necessary} that our sufficient condition in \cref{thm:symEF1} is not a necessary one,
we do not focus on the vertex coloring perspective in these experiments.

\subsubsection{Finding SymEF1 Allocations via an Integer Program}
We instead formulate an integer programming model whose feasible solutions correspond to symEF1 allocations.
Let $x_{kj}$ be a binary variable representing that bundle $k \in [n]$ contains item $j \in [m]$.
Define $y_{ij\ell}$ as a binary variable denoting that item $j \in [m]$ is removed by agent $i \in [n]$ from bundle $\ell \in [n]$.
We add constraints to enforce that each item is placed in exactly one bundle, at most one item is removed from each bundle, an item can only be removed if it is in a bundle, and every agent is envy free up to one good when receiving any bundle.
Note that multiple integer-feasible solutions may exist for a given symEF1 allocation, 
because the $y$ variables do not require that the maximum-valued item must be removed.
\begin{equation*}
    \begin{aligned}
            &\sum_{k \in [n]} x_{kj} = 1
                 &\quad&\text{for all $j \in [m]$}
                 &\quad\quad&{\text{(one bundle per item)}}
            \\
            &
                \sum_{j \in [m]} y_{ij\ell} \le 1
                 &\quad&\text{for all $i,\ell \in [n]$}
                 &\quad\quad&{\text{(agent $i$ can remove}}
                 \\[-15pt]
                & && &&{\text{one item from bundle $\ell$)}}
            \\
            &%
                y_{ij\ell} \le x_{\ell j}
                 &\quad&\text{for all $i,\ell \in [n]$, $j \in [m]$}
                 &\quad\quad&{\text{(can only remove}}\phantom{\sum_{j\in[m]}}
                 \\[-15pt]
                    & && &&{\text{if item is present)}}
            \\
            &\sum_{j \in [m]} \itemvalue{i}{j} x_{kj}
            \ge \sum_{j \in [m]} \itemvalue{i}{j} (x_{\ell j} - y_{ij\ell})
                &\quad&\text{for all $i,k,\ell \in [n]$, $k \ne \ell$} 
                    &\quad\quad&{\text{(agent $i$ is EF1}}
                    \\[-15pt]
                    & && &&{\text{with bundle $k$)}}
            \\
            &x \in \{0,1\}^{n \times m},\, y \in \{0,1\}^{n \times m \times n}.
    \end{aligned}
\end{equation*}

\subsubsection{Finding SymEF1 Allocations via a Heuristic}
As integer programming also belongs to the class of $\mathcal{NP}$-hard problems,
we will only resort to the integer program if the following greedy heuristic summarized in \cref{alg:symef1-heuristic} fails to construct a symEF1 allocation.
The heuristic maintains a partial symEF1 allocation and iteratively extends it.
At the end, either all items are assigned to bundles, or the heuristic returns ``symEF1 allocation not found''.
We consider three simple cases for extending a partial allocation with an item $j$ not currently in a bundle:
\begin{enumerate}[label={\textbf{Case \arabic*}:},ref=\arabic*,leftmargin=*]
    \item \label{case1} add item $j$ to an existing bundle;
    \item \label{case2} add item $j$ to a bundle after moving one item to another bundle;
    \item \label{case3} add item $j$ to a bundle after swapping items between two bundles.
\end{enumerate}
We track how often each of these cases succeeds in our experiments, along with the number of times we resort to the integer program.
This suggests a notion of instance difficulty that may be of independent interest.

\begin{algorithm}[t]
\caption{Heuristic for finding a symEF1 allocation.}
\label{alg:symef1-heuristic}
\algrenewcommand\algorithmicindent{1.0em}%
\begin{algorithmic}[1]
\Require Set of valuations $\{\itemvalue{i}{j} \suchthat i \in [n],\, j \in [m]\}$ by $n$ agents for $m$ items
\Ensure A symEF1 allocation for $n$ agents, or ``symEF1 allocation not found''

\For{$k \in [n]$}
\Comment{Create initial empty symEF1 allocation}
    \State $\bundle{k} \gets \emptyset$
\EndFor
\State $\mathcal{J} \gets [m]$ \Comment{Set of unallocated items}
\State $\itemallocated \gets \True$

\While{$\mathcal{J} \ne \emptyset$ and $\itemallocated$}
    \State $\itemallocated \gets \False$
        \Comment{Stop if no more items allocated}
    \For{$j \in \mathcal{J}$}

        \For{$k \in [n]$}
        \Comment{\underline{Case~\ref{case1}}: Add item $j$ to a bundle $k$}
            \State $\bundle{k} \gets \bundle{k} \cup \{j\}$
            \State
                \textbf{if} $(\bundle{1},\ldots,\bundle{n})$ is symEF1,
                set $\itemallocated \gets \True$,
                $\mathcal{J} \gets \mathcal{J} \setminus \{j\}$,
                and \textbf{break}
            \State
                \textbf{else} $\bundle{k} \gets \bundle{k} \setminus \{j\}$
                \Comment{Restore original bundles}
        \EndFor

        \State \textbf{if} $\itemallocated$, \textbf{continue}

        \For{$k, \ell \in [n], j_{k} \in \bundle{k}$}
        \Comment{\underline{Case~\ref{case2}}: Move item $j_{k}$ from bundle $k$ to $\ell$, add item $j$}
                \State
                    $\bundle{k} \gets \bundle{k} \cup \{j\} \setminus \{j_{k}\}$;
                    $\bundle{\ell} \gets \bundle{\ell} \cup \{j_{k}\}$
                \State
                    \textbf{if} $(\bundle{1},\ldots,\bundle{n})$ is symEF1,
                    set $\itemallocated \gets \True$,
                    $\mathcal{J} \gets \mathcal{J} \setminus \{j\}$,
                    and \textbf{break}
            \State \textbf{else}
                $\bundle{k} \gets \bundle{k} \cup \{j_{k}\} \setminus \{j\}$;
                $\bundle{\ell} \gets \bundle{\ell} \setminus \{j_{k}\}$
                \Comment{Restore original bundles}
        \EndFor

        \State \textbf{if} $\itemallocated$, \textbf{continue}

        \For{$k, \ell \in [n], j_{k} \in \bundle{k}, j_{\ell} \in \bundle{\ell}$}
        \Comment{\underline{Case~\ref{case3}}: Swap items $j_{k}$ and $j_{\ell}$, add item $j$}
            \State
                    $\bundle{k} \gets \bundle{k} \cup \{j, j_{\ell}\} \setminus \{j_{k}\}$;
                    $\bundle{\ell} \gets \bundle{\ell} \cup \{j_{k}\} \setminus \{j_{\ell}\}$
                \State
                    \textbf{if} $(\bundle{1},\ldots,\bundle{n})$ is symEF1,
                    set $\itemallocated \gets \True$,
                    $\mathcal{J} \gets \mathcal{J} \setminus \{j\}$,
                    and \textbf{break}
            \State \textbf{else}
                $\bundle{k} \gets \bundle{k} \cup \{j_{k}\} \setminus \{j, j_{\ell}\}$;
                $\bundle{\ell} \gets \bundle{\ell} \cup \{j_{\ell}\} \setminus \{j_{k}\}$
                \Comment{Restore original bundles}
        \EndFor
    \EndFor
\EndWhile

\State \textbf{if} $\mathcal{J} = \emptyset$, \Return $(\bundle{1},\ldots,\bundle{n}\}$; \textbf{else} \Return ``symEF1 allocation not found''

\end{algorithmic}
\end{algorithm}

\subsection{Computational Setup}
\label{sec:simulation:computational-setup}
Our code is implemented in Julia~\citepar{Julia-2017}, run with version 1.10.0 (2023-12-25),
and executed on a workstation equipped with an i9-13900K processor and 128 GB DDR5 RAM.
All integer programs are solved with Gurobi~\citepar{Gurobi10.0.3}.
Gurobi is given access to all 32 threads on the machine, which includes both performance and efficiency cores and hence affects our reports of wall clock running times.

For a fixed number of agents $n$ and items $m$, we generate an valuation matrix uniformly at random with integer entries in $\{0,1,\ldots,M\}$, where the maximum item value $M$ is an input parameter that represents different ways that valuations might be elicited.
We test $n \in \{3,4,5\}$, $m \in \{5,6,7,8,9,10,15\}$, and $M \in \{10,10^2,10^3,10^4\}$.
The reported results are averaged over 100,000 valuation matrices for $n \le 4$ and 10,000 replications for $n=5$.

\subsection{Effectiveness of Heuristic Algorithm}
\label{sec:heuristic}
In our early experiments, we solved the integer programming model for symEF1 allocations for every instance.
This was exorbitantly computationally costly, and limited the scope of our investigation.
We developed the greedy algorithm with this motivation, and in this section, we first explore how successful is this heuristic.

\subsubsection{Success of Heuristic for Two Agents}
We begin with an analysis for $n=2$ agents.
By \cref{cor:n2-exists}, a symEF1 allocation always exists for two agents and can be efficiently found without resorting to the integer program, by finding a bipartition of an associated item graph.
However, our question is:
    how often does \cref{alg:symef1-heuristic} fail to find a symEF1 allocation?
Curiously, the answer---in our limited simulations---is \emph{never}.
Specifically, we run an additional set of experiments with $n=2$, $m \in \{5,\ldots,100\}$, $M = 10^{4}$, and $10^{5}$ replications.
On average across all values of $m$, around 92.1\% of items were added to the outputted symEF1 allocation using Case~\ref{case1}, 7.8\% using Case~\ref{case2}, and the remainder with Case~\ref{case3}, with no integer programming calls required.
Furthermore, the percent of items allocated with Case~\ref{case1} increases with $m$: around 90.2\% of the time when averaging across $m \in [5,10)$, 90.7\% for $m \in [10,20)$, 92.5\% for $m \in [50,60)$, and 93.1\% for $m \in [90,100]$.
The running time for the heuristic across all replications is under 2 seconds when $m \le 10$ and increases to 30 seconds for $m = 100$.

Based on these results, one might naturally conjecture that the greedy heuristic suffices to find symEF1 allocations for all instances when there are two agents.
Unfortunately, this is not the case, as we show next in \cref{exp:2swap-not-enough},
in which a partial symEF1 allocation of eight items cannot be extended to a ninth item, even using the swaps of Case~\ref{case3}.

\begin{example}
\label{exp:2swap-not-enough}
Consider a two-agent nine-item instance with item valuations given in \cref{fig:2-swaps-exp}.
Suppose that our heuristic algorithm leads to the following partial symEF1 allocation of the first eight items:
    $\alloc = (\bundle{1},\bundle{2})$,
where $\bundle{1} = \{a,b,c,d\}$
and $\bundle{2} = \{e,f,g,h\}$.
\begin{table}
    \centering
    \caption{A set of valuations that requires more than one swap of items to allocate a new item $j$ when starting with the symEF1 allocation indicated by boxes around the first 8 items.}
    \label{fig:2-swaps-exp}
    	\begin{tabular}{lccccccccc}
		\toprule
	 			& Item $a$ & Item $b$ & Item $c$ & Item $d$ & Item $e$ & Item $f$ & Item $g$ & Item $h$ & Item $j$\\
		\midrule
	 	Agent~\player{1} & \toOne{40} & \toOne{40} & \toOne{40} & \toOne{36} & 33 & 33 & 33 & 33 & 32 \\
		Agent~\player{2} & 33 & 33 & 33 & 33 & \toTwo{36} & \toTwo{40} & \toTwo{40} & \toTwo{40} & 32 \\
        \bottomrule
	   \end{tabular}
\end{table}

Now suppose we wish to allocate the ninth item $j$.
First, we verify that $\alloc$ is symEF1
(since agent 1 and agent 2 are symmetric, we only check for agent 1):
\begin{equation*}
    v_{1}(\bundle{2}) \ge v_{1}(\bundle{1}) - \bar{v}_{1}(\bundle{1}) = 132 \geq 156 - 40 = 116.
\end{equation*}
and $v_{1}(\bundle{1}) > v_{1}(\bundle{2})$.

Next, we see that we cannot add item $j$ to either bundle (Case~\ref{case1}), because
\begin{equation*}
    v_{1}(A_{2}) = 132 < 148 = v_{1}(A_{1}) + v_{1j} - \bar{v}_{1}(A_{1}).
\end{equation*}
We cannot add $j$ to $A_{2}$ because the situation is mirrored for agent 2.

It is not possible for us to move a single item from $A_{2}$ and allocate $j$ to $A_{2}$ (Case~\ref{case2}), because
\begin{equation*}
    v_{1}(A_{2}) - v_{1e} + v_{1j} = 131 < v_{1}(A_{1}) + v_{1e} - \maxItemValAgentBundle{1}{A_{1}}.
\end{equation*}
Additionally, we cannot move one item from $A_{1}$ and allocate $j$ to $A_{1}$ because of the symmetry with agent 2.

Assume that we swap items $b$ and $f$ (Case~\ref{case3}), this is the largest valued swap possible in terms of the difference between the items values for each agent. The resulting allocation of this swap is
\begin{equation*}
    A_{1}^{1} = \{a,c,d,f\} \text{ and } A_{2}^{1} = \{b,e,g,h\}.
\end{equation*}
We still cannot allocate item $j$ because
\begin{equation*}
    v_{1}(A_{1}^1) - \maxItemValAgentBundle{1}{A_{1}^1} + v_{1j} = 141 > 139 = v_{1}(A_{2}^1).
\end{equation*}
Therefore, we must perform more than a single swap of items to allocate item $j$. If we swap items $d$ and $e$ we get the following allocation
\begin{equation*}
    A_{1}^{2} = \{a,c,e,f\} \text{ and } A_{2}^{2} = \{b,d,g,h\}.
\end{equation*}
Allocating item $j$ to $A_{2}^2$ results in a symEF1 allocation. \exqed
\end{example}

Despite this example, the simulation results do suggest that the greedy heuristic is effective when valuations are drawn uniformly at random, and that there is a low density of instances in which a swap (Case~\ref{case3}) does not suffice to produce a symEF1 allocation.

\subsubsection{Need for Integer Programming Subroutine}
Next, we turn to the case of three or more agents.
For ease, we focus on the runs with $M = 10^{4}$ as the maximum item value, as this setting of $M$ represents the hardest instances for a fixed $n$ and $m$.
\cref{tab:stats} reports several relevant statistics:
the first column is the number of agents $n$;
the second column is the number of items $m$;
the third column is the percent of the replications for which a symEF1 allocation was found;
the next three columns give the percent of the items allocated for each of the three cases in \cref{alg:symef1-heuristic}, where the percent is among those instances in which the heuristic found a symEF1 allocation;
the seventh column in turn gives the percent of instances in which the heuristic failed to find a symEF1 allocation, and so the integer programming subroutine was called,
and the eight column shows the total wall clock running time across all replications for each $(n,m)$ parameter combination.
While we include the ``\% symEF1'' column, we reserve our discussion of it to the next subsections.

{
\sisetup{
	group-separator = {,},
	group-minimum-digits = 4,
    table-format = 2.2,
    round-mode=places,
    round-precision=2,
}
\begin{table}
    \centering
    \caption{Statistics for three and four agents on percent symEF1 instances, heuristic algorithm performance, percent integer programming calls, and running time.}
    \label{tab:stats}
    \begin{tabular}{@{} 
        S[table-format=1.0] %
        S[table-format=2.0] %
        S[table-format=3.2] %
        *{2}{S[table-format=2.2]} %
        S[table-format=1.2] %
        S[table-format=2.2] %
        S[table-format=4.1,round-precision=1] %
        @{}}
        \toprule
         {$n$} & {$m$} & {\% symEF1} & {\% Case~1} & {\% Case~2} & {\% Case~3} & {\% IP} & {Time (s)} \\
         \midrule
            3 & 5  & 78.503 & 83.6423934 & 16.3576066 & 0.          & 36.075 & 174.716319 \\
3 & 6  & 98.619 & 80.9999064 & 18.944796  & 0.05529765 & 25.253 & 50.2214305 \\
3 & 7  & 99.656 & 78.1276074 & 21.3513466 & 0.521046   & 19.53  & 47.8894365 \\
3 & 8  & 99.952 & 76.4384862 & 22.6816645 & 0.87984933 & 15.71  & 39.5164564 \\
3 & 9  & 100.    & 75.3969806 & 23.6144504 & 0.98856895 & 13.073 & 32.8472994 \\
3 & 10 & 100.    & 74.5352707 & 24.4182846 & 1.04644473 & 11.357 & 33.1651054 \\
3 & 15 & 100.    & 72.7833533 & 26.197175  & 1.01947172 & 7.233  & 35.9135798 \\
\midrule
 \multicolumn{3}{@{}l}{{Average}}          & 77.417714  & 21.9379034 & 0.64438263 & 18.3187143 & \\
  \midrule
4 & 5  & 23.053 & 88.4752527 & 11.5247473 & 0.          & 76.947     & 2150.25664 \\
4 & 6  & 25.129 & 82.8334782 & 17.1665218 & 0.          & 86.196     & 3688.48959 \\
4 & 7  & 48.248 & 78.9280296 & 20.9911933 & 0.08077704 & 85.498     & 3901.20189 \\
4 & 8  & 90.863 & 75.7993241 & 23.9290356 & 0.27164024 & 80.765     & 890.592803 \\
4 & 9  & 97.643 & 71.9662892 & 26.6950567 & 1.33865411 & 76.361     & 727.14849  \\
4 & 10 & 99.219 & 69.0882255 & 28.364267  & 2.54750747 & 72.899     & 779.073858 \\
4 & 15 & 99.998 & 60.4973526 & 34.1148655 & 5.38778196 & 56.436     & 722.802731 \\
\midrule
  \multicolumn{3}{@{}l}{{Average}}           & 75.3697074 & 23.2550982 & 1.3751944  & 76.4431429 & \\
\midrule
5 & 6  & 5.87  & 89.4094265 & 10.5905735 & 0.          & 94.13      & 941.637128 \\
5 & 7  & 3.41  & 83.7714286 & 16.2285714 & 0.          & 98.75      & 1607.03925 \\
5 & 8  & 6.68  & 80.1724138 & 19.8275862 & 0.          & 99.42      & 2083.34675 \\
5 & 9  & 21.67 & 76.8518519 & 22.9938272 & 0.15432099 & 99.28      & 952.766684 \\
5 & 10 & 70.99 & 74.3157895 & 25.1578947 & 0.52631579 & 99.05      & 568.762863 \\
5 & 15 & 100.   & 59.0332326 & 32.5276939 & 8.43907351 & 96.69      & 416.097824 \\
\midrule
  \multicolumn{3}{@{}l}{{Average}}           & 77.2590238 & 21.2210245 & 1.51995172 & 97.8866667 &            \\
         \bottomrule
    \end{tabular}
\end{table}
}

From this table, we see that among instances in which \cref{alg:symef1-heuristic} succeeds for $m = 15$:
for $n=3$,
    around 73\% of the items are allocated via Case~\ref{case1},
        26\% via Case~\ref{case2},
        and 1\% via Case~\ref{case3}.
The proportion of successful item placements via Case~\ref{case1} decreases to 60\% for $n=4$ and $n=5$,
with a relatively higher reliance on Case~\ref{case3} for $n=5$, suggesting the heuristic is less effective for these instances.
Correspondingly, we see in column ``\% IP'' that in fact, for $n=5$, nearly all (over 96\% for $m=15$) of the instances required the integer program.
This is why, for $n = 5$, we attempt an order of magnitude fewer replications ($10^4$ instead of $10^5$).
However, we do not currently have a simple characterization of instances for which the heuristic will fail,
and whether such instances are indeed more challenging by some metric.

Next, we examine the effect of the number of items $m$.
As seen in the last column, and quite obviously, running time tends to correlate with the percent of integer programming subroutine invocations (that is, how often the heuristic algorithm fails to find a symEF1 allocation).
Moreover, there is a dramatic increase in running time if there is a decrease in the percent of instances in which a symEF1 allocation exists; this is because the integer program will terminate once a feasible solution is found, and proving the lack of a symEF1 solution is more computationally expensive.
In terms of the heuristic, the effectiveness of Case~\ref{case1} decreases as $m$ increases, which is sensible as the first $n$ items can always be placed in this way.
This is the opposite of the effectiveness we observed of this case as a function of $m$ with $n=2$ agents, but we conjecture that for large enough $m$, the phenomenon would occur for any number of agents,
because the instances appear to become easier (from the perspective of finding symEF1 allocations) as $m$ grows.

\subsection{Effect of Number of Items}
\label{sec:simulation:num-items}

To further support our sense that instances become easier as the number of items increases,
in \cref{fig:effect-of-num-items}, we plot the proportion of instances that are symEF1 for every combination of parameters $n$ (number of agents) and $M$ (maximum item value), as a function of the number of items $m$ on the horizontal axis.
As observed in \cref{tab:stats}, by $m=15$ items, \emph{all} instances (in our simulation) have symEF1 allocations.
For $n=5$ agents, the growth is nonmonotonic: there is first a decrease in the proportion of symEF1 instances when $m=7$ compared to $m=6$.
The reason is that, intuitively, when $m$ is close to $n$, most instances will be symEF1, e.g., for 100\% of valuations when $m \le n$.
However, perhaps surprisingly, the instances with fewest symEF1 allocations appear when $m$ is relatively close to $n$.

\begin{figure}
    \centering
    \includegraphics[width=.75\textwidth]{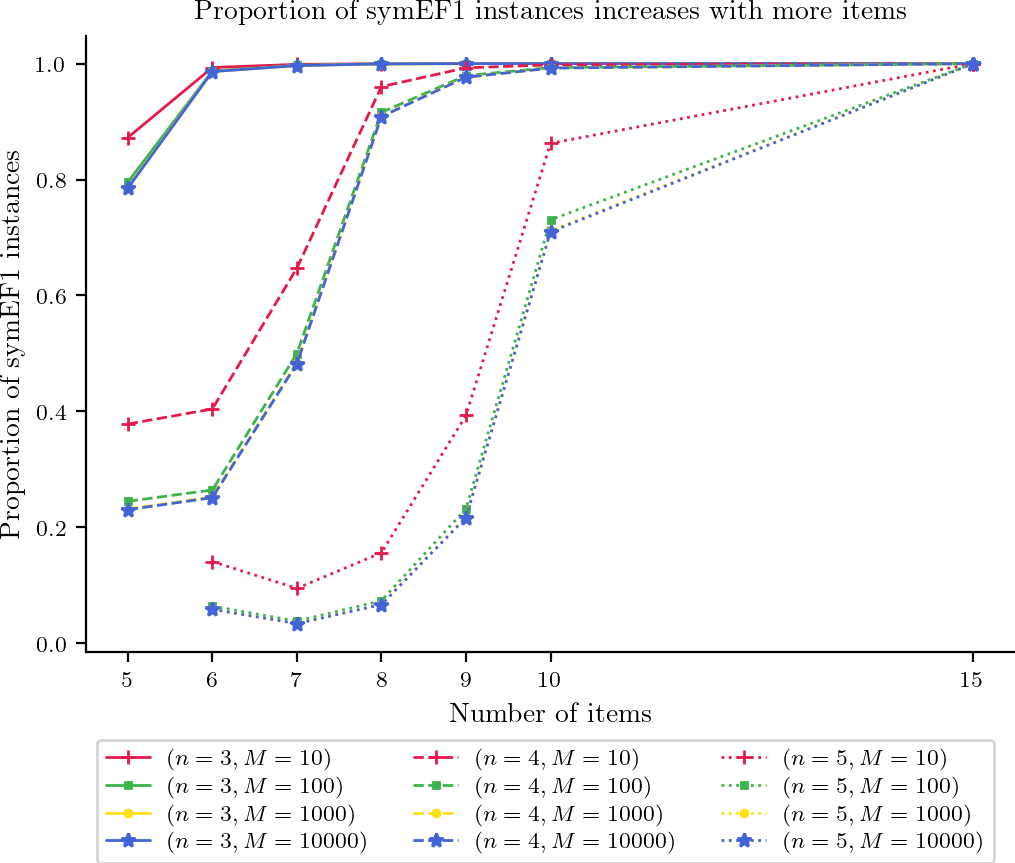}
    \caption{Proportion of instances with symEF1 allocations quickly tends to 1 with more items.}
    \label{fig:effect-of-num-items}
\end{figure}

\subsection{Effect of Maximum Item Value}
\label{sec:simulation:max-item-value}

\begin{figure}
    \centering
    \includegraphics[width=.75\textwidth]{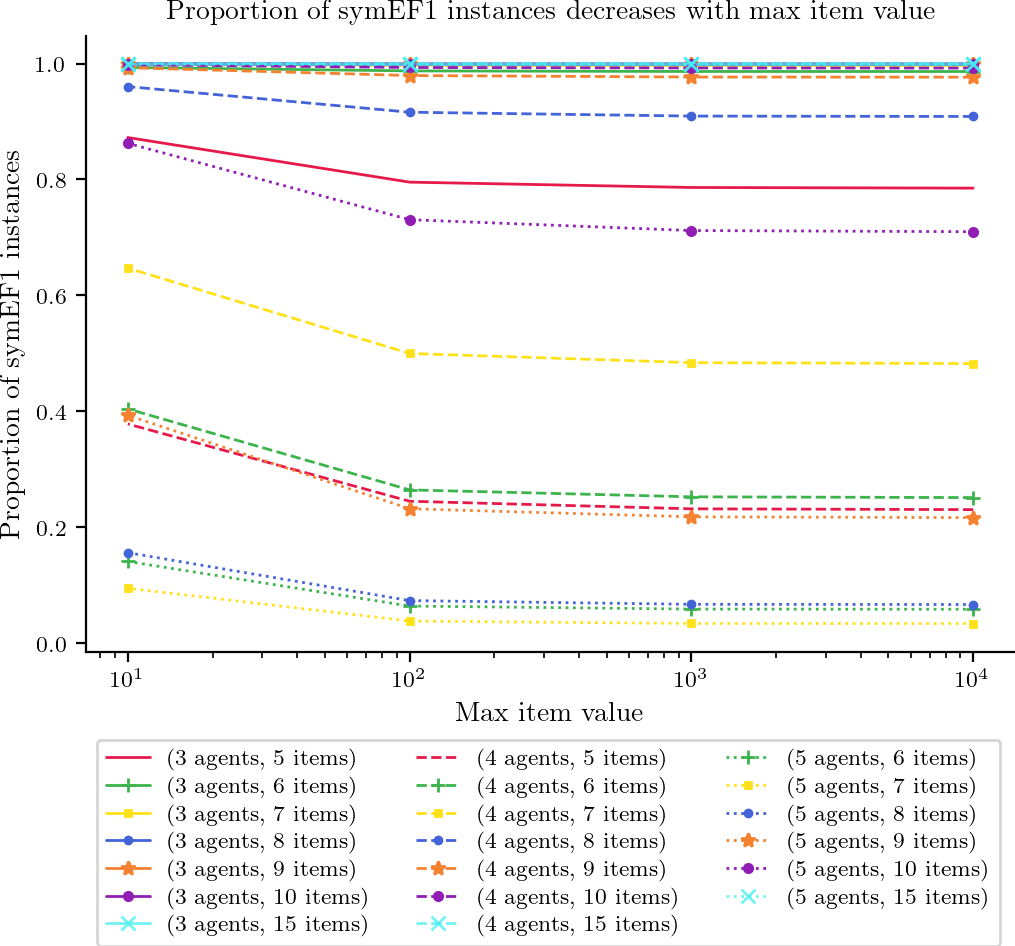}
    \caption{The proportion of instances with symEF1 allocations decreases with higher granularity of valuations.}
    \label{fig:effect-of-max-item-value}
\end{figure}

We end with a short summary of the effect of the maximum item value, $M$.
As seen in \cref{fig:effect-of-max-item-value}, there are fewer symEF1 instances when $M$ is larger.
More concretely, changing $M$ from $10$ to $100$ seems to have a large impact, while further increases in $M$ do not significantly impact the results.
This may inspire---though our results certainly do not directly imply it---a policy implication that, when eliciting agents' valuations, a rating scale of $1$--$10$ may be too coarse for some applications.

\section{Conclusion}
\label{sec:conclusion}

In \cref{sec:coloring-graph}, we proved that a symEF1 allocation always exists for two agents. In \cref{sec:2agents-4items}, we showed that some valuations, regardless of the number of items, may only give a single symEF1 allocation by our sufficient condition. However, we conjecture that multiple distinct symEF1 allocations always exist.

\begin{conjecture}\label{con:2-symef1-exists}
    If there exists a symEF1 allocation for $n$ agents and $m > n$ distinct items, then there are at least two distinct symEF1 allocations.
\end{conjecture}

The simulation results in \cref{sec:simulation} indicate that the proportion of symEF1 instances increases rapidly with the number of items. This fuels speculation for our second conjecture about the abundence of both symEF1 instances as well as symEF1 allocations.

\begin{conjecture}
    For any fixed number of agents $n$, as the number of items $m$ grows to infinity, the probability that an instance has a symEF1 allocation goes to $1$.
\end{conjecture}

Though we proved a sufficient condition for the existence of a symEF1 allocation in \cref{sec:separating-tuples}, a necessary condition for a symEF1 allocation has eluded us.
We also believe that weaker symmetric fairness notions may always be satisfiable.
Specifically, an EF$k$ allocation allows an agent to remove up to $k$ items from an envied bundle, so a symEF$k$ solution is easier to find when $k$ is larger.
We believe that a sufficient condition for existence of a symEF$k$ allocation is an especially promising direction and offer a final conjecture.

\begin{conjecture}
    For any instance with $n$ agents and $m$ items, a symEF$(n-1)$ allocation exists.
\end{conjecture}

Note that, since $m$ can be much larger than $n$, removing $n-1$ items may be only a small fraction of all goods.

\ifanon
\else
\begin{acks}
    We appreciate the input and time provided by
    T.J. Jefferson and Kayla Oates,
    who, as undergraduate students at the University of Florida, delved into an early proof attempt for the two-agent case.
    The idea of pursuing multiple symEF1 allocations is due to a conjecture by Yu Yang.
\end{acks}
\fi

\biblio{}

\appendix

\section{Proof of \texorpdfstring{\cref{prop-no-unique-2}}{Proposition~17}}
\label{app:proofs}

\begingroup
\def\thetheorem{\ref{prop-no-unique-2}}
\begin{proposition}
    Any instance with two agents and four distinct items admits at least two distinct symEF1 allocations.
\end{proposition}
\addtocounter{theorem}{-1}
\endgroup
\begin{proof} %
For compactness and readability, we assign each possible allocation a number.
\begin{alignat}{2}
    &A_{1} = \{a\},     &\quad&A_{2} = \{b,c,d\} \tag{\ensuremath{\mathcal{A}_1}}\label{m4lb-alloc1} \\
    &A_{1} = \{b\},     &&A_{2} = \{a,c,d\} \tag{\ensuremath{\mathcal{A}_2}}\label{m4lb-alloc2} \\
    &A_{1} = \{c\},     &&A_{2} = \{a,b,d\} \tag{\ensuremath{\mathcal{A}_3}}\label{m4lb-alloc3} \\
    &A_{1} = \{d\},     &&A_{2} = \{a,b,c\} \tag{\ensuremath{\mathcal{A}_4}}\label{m4lb-alloc4} \\
    &A_{1} = \{a,b\},   &&A_{2} = \{c,d\} \tag{\ensuremath{\mathcal{A}_5}}\label{m4lb-alloc5} \\
    &A_{1} = \{a,c\},   &&A_{2} = \{b,d\} \tag{\ensuremath{\mathcal{A}_6}}\label{m4lb-alloc6} \\
    &A_{1} = \{a,d\},   &&A_{2} = \{b,c\} \tag{\ensuremath{\mathcal{A}_7}}\label{m4lb-alloc7}
\end{alignat}
We fix an ordering for agent 1's values: $v_{1a} > v_{1b} > v_{1c} > v_{1d}$. We assume unique item valuations and for convenience we define $\mathcal{A}_{i} \defeq (A_{1}^i, A_{2}^i)$. For this ordering, allocations \ref{m4lb-alloc3} and \ref{m4lb-alloc4} are not symEF1 because
\begin{equation*}
    v_{1}(A_{1}^3) = v_{1c} < v_{1b} + v_{1d} = v_{1}(A_{2}^3) - \maxItemValAgentBundle{1}{A_{2}^3}
\end{equation*}
and
\begin{equation*}
    v_{1}(A_{1}^4) = v_{1d} < v_{1b} + v_{1c} = v_{1}(A_{2}^4) - \maxItemValAgentBundle{1}{A_{2}^4}.
\end{equation*}
Also, for this ordering, allocations \ref{m4lb-alloc6} and \ref{m4lb-alloc7} are guaranteed to be symEF1 for agent 1 because
\begin{equation*}
    v_{1}(A_{1}^6) = v_{1a} + v_{1c} > v_{1d} = v_{1}(A_{2}^6) - \maxItemValAgentBundle{1}{A_{2}^6}
\end{equation*}
\begin{equation*}
    v_{1}(A_{2}^6) = v_{1b} + v_{1d} > v_{1c} = v_{1}(A_{1}^6) - \maxItemValAgentBundle{1}{A_{1}^6}
\end{equation*}
and
\begin{equation*}
    v_{1}(A_{1}^7) = v_{1a} + v_{1d} > v_{1c} = v_{1}(A_{2}^7) - \maxItemValAgentBundle{1}{A_{2}^7}
\end{equation*}
\begin{equation*}
    v_{1}(A_{2}^7) = v_{1b} + v_{1c} > v_{1d} = v_{1}(A_{1}^7) - \maxItemValAgentBundle{1}{A_{1}^7}
\end{equation*}
If $v_{1b} > v_{1c} + v_{1d}$, then allocation \ref{m4lb-alloc1} and \ref{m4lb-alloc2} are also symEF1 for agent 1 otherwise $v_{1b} \le v_{1c} + v_{1d}$ implies that allocation \ref{m4lb-alloc5} is viable for agent 1.

Thus the set of viable allocations is $\{\ref{m4lb-alloc5},\ref{m4lb-alloc6},\ref{m4lb-alloc7}\}$ or $\{\ref{m4lb-alloc1}, \ref{m4lb-alloc2}, \ref{m4lb-alloc6}, \ref{m4lb-alloc7}\}$. There is the possibility that $v_{1a} > v_{1c} + v_{1d}$ and $v_{1b} \leq v_{1c} + v_{1d}$. This gives us the viable allocations $\{\ref{m4lb-alloc1}, \ref{m4lb-alloc5}, \ref{m4lb-alloc6}, \ref{m4lb-alloc7}\}$ however we do not consider this because as we will see $\{\ref{m4lb-alloc5},\ref{m4lb-alloc6},\ref{m4lb-alloc7}\}$ is sufficient to have at least two symEF1 allocations.

We first consider the set of allocations $\{\ref{m4lb-alloc5},\ref{m4lb-alloc6},\ref{m4lb-alloc7}\}$. Notice that regardless of the ordering for agent 2 at least two of the allocations from the set must be symEF1 because they separate all of their index pairs. This is because either $(a,b), (a,c),$ or $(a,d)$ is an index pair for agent 2 and thus two of the three allocations from the set have separated this pair and are therefore symEF1 for agent 2. Notice that if the set of allocations $\{\ref{m4lb-alloc5},\ref{m4lb-alloc6},\ref{m4lb-alloc7}\}$ is symEF1 for agent 1, then we are guaranteed to have two symEF1 allocations for both agents.

If the set of symEF1 allocations for agent 1 is instead $\{\ref{m4lb-alloc1}, \ref{m4lb-alloc2}, \ref{m4lb-alloc6}, \ref{m4lb-alloc7}\}$, then there are two sets of orderings for agent 2
\begin{equation*}
    v_{2a} > v_{2c} > v_{2b} > v_{2d}
\end{equation*}
and
\begin{equation*}
    v_{2a} > v_{2d} > v_{2b} > v_{2c}
\end{equation*}
such that $\ref{m4lb-alloc5}$ and either $\ref{m4lb-alloc6}$ or $\ref{m4lb-alloc7}$ is a symEF1 allocation. Thus we have only shown that one of the allocations in $\{\ref{m4lb-alloc1}, \ref{m4lb-alloc2}, \ref{m4lb-alloc6}, \ref{m4lb-alloc7}\}$ must be symEF1 for agent 2. Next we consider $\ref{m4lb-alloc1}$ and assume that agent 2's valuations is $v_{2a} > v_{2c} > v_{2b} > v_{2d}$ as the other ordering can be proven symmetrically. This allocation is symEF1 if and only if 
\begin{equation*}
    v_{2a} \ge v_{2b} + v_{2d}
\end{equation*}
because
\begin{equation*}
    v_{2b} + v_{2c} + v_{2d} \ge 0.
\end{equation*}
Notice that if $v_{2a} < v_{2b} + v_{2d}$ and thus \ref{m4lb-alloc1} is not symEF1. Then $\ref{m4lb-alloc6}$ must be a symEF1 allocation because
\begin{equation*}
    v_{2}(A_{1}) = v_{2a} + v_{2c} \ge v_{2d} = v_{2}(A_{2}) - \maxItemValAgentBundle{2}{A_{2}}
\end{equation*}
and
\begin{equation*}
    v_{2}(A_{2}) = v_{2b} + v_{2d} \ge v_{2c} = v_{2}(A_{1}) - \maxItemValAgentBundle{2}{A_{1}}. 
\end{equation*}
Therefore either $\ref{m4lb-alloc1}$ or $\ref{m4lb-alloc6}$ must be a symEF1 allocation for agent $2$ and by the ordering we chose $\ref{m4lb-alloc7}$ is a symEF1 allocation because it separates agent 2's index pairs. Thus there are always at least two symEF1 allocations for both agents in this set. Therefore regardless of valuations for both agents there is always at least 2 symEF1 allocations when $n = 2$ and $m = 4$.
\end{proof}

\end{document}